\let\accentvec\vec  %

\documentclass[a4paper,UKenglish]{llncs}

\let\vec\accentvec

\usepackage{amsmath,amssymb,xspace}%
\usepackage[textsize=scriptsize,linecolor=black,backgroundcolor=white]{todonotes}

\usepackage{comment}
\specialcomment{shortproof}{\begin{proof}}{\end{proof}}
\specialcomment{longproof}{\begin{proof}}{\end{proof}}
\excludecomment{shortproof}

\spnewtheorem{redrule}{Reduction Rule}{\bfseries}{\itshape}
\spnewtheorem{invariant}{Invariant}{\bfseries}{\itshape}

\newcommand{\no}{\textsc{no}\xspace}
\newcommand{\yes}{\textsc{yes}\xspace}

\newcommand{\opt}{\ensuremath{\mathsf{opt}}\xspace}

\newcommand{\poly}{\ensuremath{\operatorname{poly}}}

\newcommand{\F}{\mathcal{F}}
\newcommand{\X}{\ensuremath{\mathcal{X}}\xspace}
\newcommand{\Y}{\ensuremath{\mathcal{Y}}\xspace}
\renewcommand{\P}{\mathcal{P}} %
\newcommand{\Oh}{\mathcal{O}}

\newcommand{\N}{\mathbb{N}}

\newcommand{\eref}[1]{(\ref{#1})\xspace}

\newcommand{\probnamex}[1]{\textsc{\lowercase{#1}}}
\newcommand{\probname}[1]{\probnamex{#1}\xspace}

\newcommand{\VCon}{\probnamex{Vector Connec\-ti\-vi\-ty}\xspace}
\newcommand{\VdCon}{\probnamex{Vector $d$-Connec\-ti\-vi\-ty}\xspace}
\newcommand{\VConk}{\probnamex{Vector Connec\-ti\-vi\-ty($k$)}\xspace}
\newcommand{\VdConk}{\probnamex{Vector $d$-Connecti\-vi\-ty($k$)}\xspace}

\newcommand{\HSm}{\probnamex{Hitting Set($m$)}\xspace}
\newcommand{\HS}{\probnamex{Hitting Set}\xspace}

\newcommand{\APX}{\ensuremath{\mathsf{APX}}\xspace}
\newcommand{\NP}{\ensuremath{\mathsf{NP}}\xspace}
\newcommand{\containment}{\ensuremath{\mathsf{NP\subseteq coNP/poly}}\xspace}

\newcommand{\FPT}{\ensuremath{\mathsf{FPT}}\xspace}

\newcommand{\demandiny}{\ensuremath{d^3}\xspace}

\newcommand{\req}{\ensuremath{\textit{Req}}}
\newcommand{\satcon}{\ensuremath{\textit{Sat}}}
\newcommand{\fac}{\ensuremath{\textit{Fac}}}
\newcommand{\sig}{\ensuremath{\textit{Sig}}}

\newcommand{\parproblemdef}[4]
{
\begin{quote}
#1\\
\textbf{Instance:} #2\\
\textbf{Parameter:} #3\\
\textbf{Question:} #4
\end{quote}
}

\newcommand{\problemdef}[3]
{
\begin{quote}
#1\\
\textbf{Instance:} #2\\
\textbf{Question:} #3
\end{quote}
}

\title{On Kernelization and Approximation for the Vector Connectivity Problem}

\author{Stefan Kratsch 
\and Manuel Sorge
}
\institute{Technical University Berlin, Germany, \email{$\{$stefan.kratsch$,$manuel.sorge$\}$@tu-berlin.de}}
\authorrunning{S. Kratsch and M. Sorge}

\begin{document}

\maketitle

\begin{abstract}
In the \VCon problem we are given an undirected graph $G=(V,E)$, a demand function $\phi\colon V\to\{0,\ldots,d\}$, and an integer $k$. The question is whether there exists a set $S$ of at most $k$ vertices such that every vertex $v\in V\setminus S$ has at least $\phi(v)$ vertex-disjoint paths to $S$; this abstractly captures questions about placing servers or warehouses relative to demands. The problem is \NP-hard already for instances with $d=4$ (Cicalese et al., arXiv '14), admits a log-factor approximation (Boros et al., Networks '14), and is fixed-parameter tractable in terms of~$k$ (Lokshtanov, unpublished '14).

\looseness=-1 We prove several results regarding kernelization and approximation for \VCon and the variant \VdCon where the upper bound $d$ on demands is a fixed constant. For \VdCon we give a factor $d$-approximation algorithm and construct a vertex-linear kernelization, i.e., an efficient reduction to an equivalent instance with $f(d)k=\Oh(k)$ vertices. For \VCon we have a factor $\opt$-approximation and we can show that it has no kernelization to size polynomial in $k$ or even $k+d$ unless \containment, making $f(d)\operatorname{poly}(k)$ optimal for \VdCon. Finally, we provide a write-up for fixed-parameter tractability of \VConk by giving an alternative FPT algorithm based on matroid intersection.
\end{abstract}

\section{Introduction}

In the \VCon problem we are given an undirected graph $G=(V,E)$, a demand function $\phi\colon V\to\{0,\ldots,d\}$, and an integer $k\in\N$. The question is whether there exists a set $S$ of at most $k$ vertices of $G$ such that every vertex $v\in V$ is either in $S$ or has at least $\phi(v)$ vertex-disjoint paths to vertices in $S$ (where the paths may share the vertex $v$ itself).

\problemdef{\VCon}{A graph $G=(V,E)$, a function~$\phi\colon V\to\{0,\ldots,d\}$, $k\in\N$.}{Is there a set $S$ of at most $k$ vertices such that each vertex $v\in V\setminus S$ has $\phi(v)$ vertex-disjoint paths with endpoints in $S$?}
The value~$\phi(v)$ is also called the \emph{demand} of vertex~$v$. We call~$S\subseteq V$ a \emph{vector connectivity set} for~$(G, \phi)$, if it fulfills the requirements above.
We do not formally distinguish decision and optimization version; $k$ is not needed for the latter.

\looseness=-1 For intuition about the problem formulation and applications one may imagine a logistical problem about placing warehouses to service locations on a map (or servers in a network): Each location has a particular demand, which could capture volume and redundancy requirements or level of importance. Demand can be satisfied by placing a warehouse at the location or by ensuring that there are enough disjoint paths from the location to warehouses; both can be seen as guaranteeing sufficient throughput or ensuring access that is failsafe up to demand minus one interruptions in connections. In this way, \VCon can also be seen as a variant of the \probname{Facility Location} problem~\cite{HamacherD02} in which the costs of serving demand are negligible, but instead redundancy is required.

\emph{\textbf{Related work.}} The study of the \VCon problem was initiated recently by Boros et al.~\cite{BorosHHM14} who gave polynomial-time algorithms for trees, cographs, and split graphs. Moreover, they obtained an $\ln n+2$-factor approximation algorithm for the general case.
More recently, Cicalese et al.~\cite{CicaleseMR14} continued the study of \VCon and amongst other results proved that it is \APX-hard (and \NP-hard) on general graphs, even when all demands are upper bounded by four. 
In a related talk during a recent Dagstuhl seminar\footnote{Dagstuhl Seminar 14071 on ``Graph Modification Problems.''} Milani\v{c} asked whether \VCon is fixed-parameter tractable with respect to the maximum solution size~$k$. This was answered affirmatively by Lokshtanov.\footnote{Unpublished result.} %

\emph{\textbf{Our results.}} We obtain results regarding kernelization and approximation for \VCon and \VdCon where the maximum demand $d$ is a fixed constant. We also provide a self-contained write-up for fixed-parameter tractability of \VCon with parameter $k$ (Sec.~\ref{section:fpt}); it is different from Lokshtanov's approach and instead relies on a matroid intersection algorithm of Marx~\cite{Marx09}.

Our analysis of the problem starts with a new data reduction rule stating that we can safely ``forget'' the demand of $r:=\phi(v)$ at a vertex $v$ if $v$ has vertex-disjoint paths to $r$ vertices each of demand at least $r$ (Sec.~\ref{section:reductionrule}). After exhaustive application of the rule, all remaining vertices of demand $r$ must have cuts of size at most $r-1$ separating them from other vertices of demand at least $r$. By analyzing these cuts we then show that any yes-instance of \VdCon can have at most $d^2k$ vertices with nonzero demand; the corresponding bound for \VdCon is $k^3+k$. Both bounds also hold when replacing $k$ by the optimum cost $\opt$. This directly would yield factor $d^2$ and factor $\opt^2+1$ approximation algorithms. We improve upon this in Sec.~\ref{section:approximation} by giving a variant of the reduction rule that works correctly relative to any partial solution $S_0$, which can then be applied in each round of our approximation algorithm. The algorithm follows the local-ratio paradigm and, surprisingly perhaps, proceeds by always selecting a vertex of \emph{lowest} demand (plus its cut). We thus obtain ratios of $d$ and $\opt$ respectively, i.e., the returned solution is of size at most $d\cdot \opt$ for \VdCon and at most $\opt^2$ for \VCon.

Regarding kernelization we show in Sec.~\ref{section:kernellowerbound} that there is no kernel with size polynomial in $k$ or even $k+d$ for \VCon unless \containment (and the polynomial hierarchy collapses); our proof also implies that the problem is $\mathsf{WK[1]}$-hard (cf.~\cite{HermelinKSWW13}). Nevertheless, when $d$ is a fixed constant, we prove that a vertex-linear kernelization is possible. A non-constructive proof of this can be pieced together from recent work on meta kernelization on sparse graph classes (see below). Instead we give a constructive algorithm building on an explicit (though technical) description of what constitutes equivalent subproblems. We also have a direct proof for the number of such subproblems in an instance with parameter $k$ rather than relying on a known argument for bounding the number of connected subgraphs of bounded size and bounded neighborhood size (via the two-families theorem of Bollobas, cf. \cite{Jukna01}). Our kernelization is presented in Sec.~\ref{section:kernelization}.
The bound of $f(d)k=\Oh(k)$ vertices is optimal in the sense that the lower bound for parameter $d+k$ rules out total size $\poly(k+d)$, i.e. we need to allow superpolynomial dependence on~$d$.

\emph{\textbf{A non-constructive kernelization argument.}} The mentioned results on meta kernelization for problems on sparse graph classes mostly rely on the notion of a protrusion, i.e., an induced subgraph of bounded treewidth such that only a bounded number of its vertices have neighbors in the rest of the graph (the \emph{boundary}). Under appropriate assumptions there are general results that allow the replacements of protrusions by equivalent ones of bounded size, which yields kernelization results assuming that the graph can be decomposed into a bounded number of protrusions. Intuitively, having small boundary size limits the interaction of a protrusion with the remaining graph. The assumption of bounded treewidth ensures fast algorithms for solving subproblems on the protrusion and also leads to fairly general arguments for obtaining small equivalent replacements. Note that for \VCon there is no reason to assume small treewidth and we also do not restrict the input graphs. Vertex-linear kernels for problems on general graphs are much more rare than the wealth of such kernels on sparse input graphs.

Arguably, the most crucial properties of a protrusion are the small boundary and the fact that we can efficiently compute subproblems on the protrusion; in principle, we do not need bounded treewidth. (Intuitively, we essentially want to know the best solution value for each choice of interaction of a global solution with the boundary of the protrusion.) Thus, it seems worthwhile and natural to define a relaxed variant of protrusions by insisting on a small boundary and efficient algorithms for solving subproblems rather than demanding bounded treewidth. Fomin et al.~\cite{FominLST13} follow this approach for problems related to picking a certain subset of vertices like \probname{Dominating Set}: Their relaxed definition of a $r$-DS-protrusion requires a boundary of size at most $r$ and the existence of a solution of size at most $r$ for the protrusion itself; the latter part implies efficient algorithms since we can afford to simply try all $r=\Oh(1)$ sized vertex subsets. Fomin et al.\ also derive a general protrusion replacement routine for problems that have finite integer index (a common assumption for meta kernelization, see, e.g., Fomin et al.~\cite{FominLST13}) and are monotone when provided also with a sufficiently fast algorithm, like the one implied by having a solution of size at most~$r$ for the protrusion. Fomin et al.~\cite{FominLST13} remark that the procedure is not constructive since it assumes hard-wiring an appropriate set of representative graphs, whose existence is implied by being finite integer index.

From previous work of Cicalese~\cite{CicaleseMR14} it is known that \VCon is an implicit hitting set problem where the set family consists of all connected subgraphs with neighborhood size lower than the largest demand in the set; the family is exponential size, but we get size roughly $\Oh(n^d)$ when demands are at most~$d$. The procedure of Fomin et al.~\cite{FominLST13} can be applied to minimal sets in this family and will shrink them to some size bounded by an unknown function in $d$, say $h(d)$. Then one can apply the two-families theorem of Bollobas to prove that each vertex is contained in at most $\binom{h(d)+d}{d}$ such sets. Because the solution must hit all sets with $k$ vertices, a yes-instance can have at most $k\cdot \binom{h(d)+d}{d}$ sets. This argument can be completed to a vertex-linear kernelization.

In comparison, we obtain an explicit upper bound of $d^2k\cdot 2^{d^3+d}$ for the number of subproblems that need to be replaced, by considering a set family that is different from the implicit hitting set instance (but contains supersets of all those sets). We also have a constructive description of what constitutes equivalent subproblems. This enables us to give a single algorithm that works for all values of~$d$ based on maximum flow computations (rather than requiring for each value of~$d$ an algorithm with hard-wired representative subproblems).

\section{Preliminaries}\label{section:preliminaries}

We consider only undirected, simple graphs $G=(V,E)$ where $V$ is a finite set and $E\subseteq\binom{V}{2}$; we use $V(G):=V$ and $E(G):=E$. (Here, $\binom{V}{2}$ denotes the family of all size-two subsets of~$V$.) The open neighborhood $N_G(v)$ of a vertex $v\in V(G)$ is defined by $N_G(v):=\{u\mid \{u,v\}\in E(G)\}$; the closed neighborhood is $N_G[v]:=N_G(v)\cup\{v\}$. For vertex sets $X\subseteq V(G)$ we define $N_G[X]:=\bigcup_{v\in X} N_G[v]$ and $N_G(X):=N_G[X]\setminus X$. For $X\subseteq V(G)$ we define the subgraph induced by $X$, denoted $G[X]$, as the graph $G'=(X,E')$ with $E':=\{\{u,v\}\mid \{u,v\}\in E(G) \wedge u,v\in X\}$. By $G-X$ we mean the graph $G[V(G)\setminus X]$.

Due to the nature of the \VCon problem we are frequently interested in disjoint paths from a vertex $v$ to some vertex set $S$, where the paths are vertex-disjoint except for sharing $v$. The natural counterpart, in the spirit of Menger's theorem, are \emph{$v,S$-separators} $C\subseteq V(G)\setminus\{v\}$ such that in $G-C$ no vertex of $S\setminus C$ is reachable from $v$; the maximum number of disjoint paths from $v$ to $S$ that may overlap in $v$ equals the minimum size of a $v,S$-separator. Throughout, by disjoint paths from $v$ to $S$, or $v,S$-separator (for any single vertex $v$ and any vertex set $S$) we mean the mentioned path packings and separators with special role of $v$. (In a network flow interpretation of the paths between~$v$ and~$S$, this essentially corresponds to giving $v$ unbounded capacity, whereas all other vertices have unit capacity.) Several proofs use the function $f\colon 2^V\to\N\colon U\mapsto |N(U)|$, which is well-known to be \emph{submodular}, i.e., for all $X,Y\subseteq V$ it holds that $f(X)+f(Y)\geq f(X\cap Y)+f(X\cup Y)$.

So-called \emph{closest sets} will be used frequently; these occur
naturally in cut problems but appear to have no generalized name. We
define a vertex set $C$ to be \emph{closest to $v$} if $C$ is the
unique $v,C$-separator of size at most $|C|$, where $v,C$-separator is
in the above sense. As an example, if $C$ is a minimum $s,t$-vertex
cut that, amongst such cuts, has the smallest connected component for
$s$ in $G-C$ then $C$ is also closest. The following proposition
captures some properties of closest sets, mostly specialized to
$v,S$-separators (see also~\cite{KratschW12}). %

\begin{proposition}\label{proposition:closestsets:properties}%
Let $G=(V,E)$, let $v\in V$, and let $S\subseteq V\setminus\{v\}$. The following holds.
\begin{enumerate}
\item If $C$ is a closest $v,S$-separator and $X$ the connected component of $v$ in $G-C$, then for every set $X'\subseteq X$ with $v\in X'\subsetneq X$ we have $|N(X')|>|N(X)|$.\label{claim:basic}
\item The minimum $v,S$-separator $C$ minimizing the size of the connected component of $v$ in $G - C$ is unique. The set $C$ is also the unique minimum $v,S$-separator closest to~$v$.
\item If $C_1$ and $C_2$ are minimum $v,S$-separators and $C_1$ is closest to $v$, then the connected component of $v$ in $G-C_1$ is fully contained in the connected component of $v$ in $G-C_2$.
\item If $C\subseteq V\setminus\{v\}$ is closest to $v$ then so is every subset $C'$ of $C$.
\end{enumerate}
\end{proposition}

\begin{longproof}
Let $G=(V,E)$, let $v\in V$, and let $S\subseteq V\setminus\{v\}$.
We give simple proofs for all claims; some of them use submodularity of $f\colon 2^V\to\N\colon U\mapsto |N(U)|$.
\begin{enumerate}
\item %
  Note that $N(X)=C$ or else $|N(X)|<|C|$ and $N(X)$ is a $v,C$-separator, contradicting closeness of~$C$. Let $X'\subseteq X$ with $v\in X'\subsetneq X$. Clearly $X'\cap C\subseteq X\cap C=\emptyset$, making $N(X')$ a $v,C$-separator. Since $X'\subsetneq X$ and $X$ is connected, it follows that $N(X')\cap X\neq \emptyset$ and, thus, that $N(X')\neq N(X)$. Since $C$ is closest to $v$, we must have $|N(X')|>|C|=|N(X)|$ or else $N(X')$ would be a $v,C$-separator of size at most $|C|$ but different from $C=N(X)$.

\item Assume that both $C_1$ and $C_2$ are minimum $v,S$-separators that furthermore minimize the size of the connected component of $v$ in $G-C_i$; let $X_i$ denote those components. Note that $N(X_i)=C_i$ since $C_i$ is a minimum $v,S$-separator. We have
\begin{align*}
f(X_1)+f(X_2)\geq f(X_1\cap X_2)+f(X_1\cup X_2).
\end{align*}
Clearly, $(X_1\cup X_2)\cap S=\emptyset$, implying that $N(X_1\cup X_2)$ is a $v,S$-separator. It follows that $|N(X_1)\cup N(X_2)|=f(X_1\cup X_2)\geq f(X_2)$, which implies that $f(X_1)\geq f(X_1\cap X_2)$. It is easy to see that $N(X_1\cap X_2)$ is also a $v,S$-separator of size at most $f(X_1)=|C_1|$ and, if $C_1\neq C_2$, then $X_1\neq X_2$ and $X_1\cap X_2\subsetneq X_1$; a contradiction to $X_1$ being a minimum size connected component. Thus, we have $C_1=C_2$, proving the first part of the claim.

For the second part, assume first that $C_3$ is a $v,C_1$-separator of size at most $|C_1|$ and let $X_3$ the connected component of $v$ in $G-C_3$. Since every path from $v$ to $S$ contains a vertex of $C_1$, each such path must also contain a vertex of $C_3$ (to separate $v$ from $C_1$). Thus $C_3$ is also a $v,S$-separator, but then $C_3$ is also minimum (as $|C_3|\leq |C_1|$). Note that, since $C_3$ is a $v,C_1$-separator we have $X_3\cap C_1=\emptyset$. It follows that $X_3\subseteq X_1$. If $X_3\subsetneq X_1$ then $C_3$ would violate the choice of $C_1$ as minimum $v,S$-separator with minimum component size for $v$. But then we must have $X_3=X_1$ and $C_1=C_3$, which proves closeness of $C_1$ to $v$.

Finally, assume that $C_4$ is another minimum $v,S$-separator closest to $v$; for uniqueness we want to show $C_1=C_4$. If $X_1$ and $X_4$ are the corresponding connected components, then $C_i=N(X_i)$. By submodularity of $f$ and the same arguments as above we find that $|N(X_1\cap X_4)|\leq |C_1|$. If $X_1\nsubseteq X_4$ then $v\in X_1\cap X_4\subsetneq X_1$ which, by property~\ref{claim:basic}, implies $|N(X_1\cap X_4)|>|N(X_1)|=|C_1|$; a contradiction. Else, if $X_1\subseteq X_4$, then $C_1=N(X_1)$ is a $v,C_4$-separator of size $|C_1|=|C_4|$; this implies $C_1=C_4$ by closeness of $C_4$.

\item Let $C_1$ and $C_2$ minimum $v,S$-separators and let $C_1$ be closest to $v$. Let $X_i$ the connected component of $v$ in $G-C_i$. Note that $N(X_i)=C_i$ since $C_i$ is minimum. Assume, for contradiction, that $X_1\nsubseteq X_2$, implying that $X_1\cap X_2\subsetneq X_1$. By property~\ref{claim:basic} we have $f(X_1\cap X_2)=|N(X_1 \cap X_2)| > |C_1| = |C_2|=f(X_2)$. Because
\begin{align*}
f(X_1)+f(X_2)\geq f(X_1\cap X_2)+f(X_1\cup X_2),
\end{align*}
we have $|N(X_1\cup X_2)| < f(X_1)=|C_1| = |C_2|$. However, $N(X_1 \cup X_2)$ is also a $v,S$-separator; this contradicts~$C_1$ and~$C_2$ being minimum $v,S$-separators.

\item Let $C$ closest to $v$ and let $C'\subseteq C$. If $C'$ is not closest to $v$ then there is a $v,C'$-separator $C''$ of size at most $|C'|$ and with $C''\neq C'$. Consider $\hat{C}=C''\cup(C\setminus C')$ and note that $|\hat{C}|\leq |C''|+|C\setminus C'|\leq |C'|+|C\setminus C'|=|C|$. Observe that $\hat{C}$ is a $v,C$-separator since $C''\subseteq \hat{C}$ separates $v$ from $C'$, and $C\setminus C'$ is contained in $\hat{C}$. Thus, $|\hat{C}|<|C|$ would contradict $C$ being closest to $v$ because it implies that $\hat{C}\neq C$ in addition to $\hat{C}$ being a $v,C$-separator of size at most $|C|$. Thus, $|\hat{C}|=|C|$ and by closeness of $C$ to $v$, we have $\hat{C}=C$. The latter can hold only if $C''\supseteq C'$, because $\hat{C}=C''\cup(C\setminus C')$, but then $C''=C'$ because $|C''|\leq |C'|$, contradicting $C''\neq C'$. Thus $C'$ is indeed closest to $v$.
\end{enumerate}
This completes the proof.\qed
\end{longproof}

\section{Reducing the number of demand vertices}\label{section:reductionrule}

In this section we introduce a reduction rule for \VCon that reduces the total demand. We prove that the reduction rule does not affect the solution space, which makes it applicable not only for kernelization but also for approximation and other techniques. In Section~\ref{section:kernelization}, we will use this rule in our polynomial kernelization for \VdConk. In the following two sections, as applications of these reduction rules and the insights gained we get approximation algorithms for \VdCon and \VCon, and an alternative \FPT algorithm for \VConk using a result of Marx~\cite{Marx09}.

\begin{redrule}\label{rule:reducedemand}
Let $(G,\phi,k)$ be an instance of \VCon.
If a vertex $v\in V$ has at least $\phi(v)$ vertex-disjoint paths to vertices different from $v$ with demand at least $\phi(v)$ then set the demand of $v$ to zero.
\end{redrule}

We prove that the rule does not affect the space of feasible solutions for any instance.

\begin{lemma}\label{lemma:reducedemand-safe}
Let $(G,\phi,k)$ be an instance of \VCon and let $(G,\phi',k)$ be the instance obtained via a single application of Rule~\ref{rule:reducedemand}. For every $S\subseteq V(G)$ it holds that $S$ is a solution for $(G,\phi,k)$ if and only if $S$ is a solution for $(G,\phi',k)$.
\end{lemma}

\begin{longproof}
Let $v$ denote the vertex whose demand was set to zero by the reduction rule and define $r:=\phi(v)$. Clearly, $\phi(u)=\phi'(u)$ for all vertices $u\in V(G)\setminus\{v\}$, and $\phi'(v)=0$. It suffices to show that if $S$ fulfills demands according to $\phi'$ then $S$ fulfills also demands according to $\phi$ since $\phi(u)\geq \phi'(u)$ for all $u\in V(G)$. This in turn comes down to proving that $S$ fulfills the demand of $r$ at $v$ assuming that it fulfills demands according to $\phi'$.
If $v\in S$ then the demand at $v$ is trivially fulfilled so henceforth assume that $v\notin S$.

Let $w_1,\ldots,w_r$ denote vertices different from $v$ with demand each at least $r$ such that there exist~$r$ vertex-disjoint paths from $v$ to $\{w_1,\ldots,w_r\}$, i.e., a single path to each $w_i$. Existence of such vertices is required for the application of the rule.

Assume for contradiction that $S$ does not satisfy the demand of $r$ at $v$ (recall that $v\notin S$, by assumption), i.e., that there are no $r$ vertex-disjoint paths from $v$ to $S$ that overlap only in $v$. It follows directly that there is a $v,S$-separator $C$ of size at most $r-1$. (Recall that $C$ may contain vertices of $S$ but not the vertex $v$.) Let~$R$ denote the connected component of~$v$ in~$G-C$, then the following holds for each vertex $w_i\in\{w_1,\ldots,w_r\}$:
\begin{enumerate}
 \item If $w_i\in S$ then $w_i\notin R$: Otherwise, we would have $S\cap R\supseteq\{w_i\}\neq\emptyset$ contradicting the fact that $v$ can reach no vertex of $S$ in $G-C$.
 \item If $w_i\notin S$ then $w_i\notin R$: Since $S$ fulfills demands according to $\phi'$ there must be at least $r$ vertex-disjoint paths from $w_i$ to $S$ that overlap only in $w_i$. However, since $w_i\in R$ the set $C$ is also a $w_i,S$-separator; a contradiction since $C$ has size less than $r$.
\end{enumerate}
Thus, no vertex from~$w_1,\ldots,w_r$ is contained in~$R$. This, however, implies that~$C$ separates~$v$ from~$\{w_1,\ldots,w_r\}$, contradicting the fact that there are~$r$ vertex-disjoint paths from~$v$ to~$\{w_1,\ldots,w_r\}$ that overlap only in $v$. It follows that no such $v,S$-separator~$C$ can exist, and, hence, that there are at least~$r=\phi(v)$ vertex-disjoint paths from~$v$ to~$S$, as claimed. Thus, $S$ fulfills the demand of $r$ at $v$ and hence all demand according to $\phi$. (Recall that the converse is trivial since $\phi(u)\geq \phi'(u)$ for all vertices $u\in V(G)$.)\qed
\end{longproof}

We have established that applications of Rule~\ref{rule:reducedemand} do not affect the solution space of an instance while reducing the number of vertices with nonzero demand. %

\begin{lemma}\label{lemma:reducedemand-poly}
Rule~\ref{rule:reducedemand} can be exhaustively applied in polynomial time. 
\end{lemma}

\begin{longproof}
To check whether the rule applies to some vertex~$v$ with demand~$r$ it suffices to perform one maximum flow computation for source~$v$ (with unbounded capacity) and using all vertices with demand at least~$r$ as sinks (with capacity one). Finding a vertex suitable for applying Rule~\ref{rule:reducedemand} thus takes only polynomial time and each application reduces the number of nonzero demand vertices decreases by one limiting the number of iterations to~$|V(G)|$.\qed
\end{longproof}

To analyze the impact of Rule~\ref{rule:reducedemand} we will now bound the number of nonzero demand vertices in an exhaustively reduced instance in terms of the optimum solution size $\opt$ and the maximum demand $d$. To this end, we require the following technical lemma about the structure of reduced instances as well as some notation.
If $(G,\phi,k)$ is reduced according to Rule~\ref{rule:reducedemand} then for each vertex~$v$ with demand~$r=\phi(v)\geq 1$ there is a cut set $C$ of size at most~$r-1$ that separates $v$ from all other vertices with demand at least~$r$. We fix for each vertex $v$ with demand at least one a vertex set $C$, denoted $C(v)$, by picking the unique closest minimum $v,D_v$-separator where $D_v=\{u\in V\setminus\{v\}\mid \phi(u)\geq\phi(v)\}$. Furthermore, for such vertices $v$, let $R(v)$ denote the connected component of $v$ in $G-C(v)$.

Intuitively, any solution $S$ must intersect every set $R(v)$ since $|C(v)|<\phi(v)$. The following lemma shows implicitly that Rule~\ref{rule:reducedemand} limits the amount of overlap of sets $R(v)$.

\begin{lemma}\label{lemma:cuts}
Let $(G,\phi,k)$ be reduced under Rule~\ref{rule:reducedemand}.
Let $u,v\in V(G)$ be distinct vertices with $\phi(u)=\phi(v)\geq 1$. If $R(u)\cap R(v)\neq\emptyset$ then $u\in C(v)$ or $v\in C(u)$.
\end{lemma}

\begin{longproof}
Assume for contradiction that we have $u,v$ with $\phi(u)=\phi(v)\geq 1$ and with $R(u)\cap R(v)\neq\emptyset$ and $u\notin C(v)$ and $v\notin C(u)$. We will show that this implies that at least one of $C(u)$ and $C(v)$ is not a closest minimum cut, giving a contradiction. By definition of cuts $C(u)$ and $C(v)$ as separating $u$ resp.\ $v$ from all other vertices of at least the same demand, we have $u\notin R(v)$ and $v\notin R(u)$; furthermore $u\notin C(u)$ and $v\notin C(v)$, by definition.

Let $C=C(u)\cup C(v)$ and note that $u,v\notin C$. Let $I,J$ denote the connected components of $u,v$ in $G-C$. Note that $I\subseteq R(u)$ since $C\supseteq C(u)$ and, thus, $v\notin I$. Similarly we have $u\notin J$ and thus $I$ and $J$ are two different connected components in $G-C$. As the next step, we show that
\begin{align}
|C(u)|+|C(v)|\geq |N(I)|+|N(J)|.\label{equation:improvement}
\end{align}
To this end, let us first note that $N(I)\cup N(J)\subseteq C=C(u)\cup C(v)$ by definition of $I$ and $J$ as connected components of $G-C$. Thus, every vertex $p$ that appears in exactly one of $N(I)$ and $N(J)$ contributes value one to the right-hand side of \eref{equation:improvement} and at least value one to the left-hand side since it must be contained in $C(u)\cup C(v)$. Now, for vertices $p\in N(I)\cap N(J)$ we see that they contribute value two to the right-hand side of \eref{equation:improvement}. Note that each such vertex is contained in a path from $u$ to $v$ whose other interior vertices are disjoint from $C=C(u)\cup C(v)$. Thus,~$p$ must be contained in both $C(u)$ and $C(v)$ since otherwise the corresponding set would fail to separate $u$ from $v$ (or vice versa), which is required since~$\phi(u)=\phi(v)$. Therefore, if any vertex contributes a value of two to the right-hand side, then it also contributes two to the left-hand side. This establishes Equation \eref{equation:improvement}.

Now, from \eref{equation:improvement} we immediately get that at least one of $|N(I)|\leq |C(u)|$ or $|N(J)|\leq |C(v)|$ must hold. W.l.o.g., let~$|N(I)|\leq |C(u)|$. Recall that $u\in I$. Furthermore, we can see that $I\subsetneq R(u)$ by using the fact that~$R(u)\cap R(v)\neq \emptyset$: Let $q\in R(u)\cap R(v)\subseteq R(u)$. If $q\in I$ then $u$ can reach $q$ in~$G-C$ but, in particular, also in~$G-C(v)$. By definition of~$R(v)$ and using~$q\in R(v)$, we know that~$v$ can reach~$q$ in~$G-C(v)$, implying that there is a path from~$u$ to~$v$ in~$G-C(v)$ (since there is a walk through~$q$), violating the fact that~$C(v)$ separates~$v$ from~$u$ (amongst others). Thus~$q\notin I$ and since $I\subseteq R(u)$ follows from $C\supseteq C(u)$, we get that~$I\subsetneq R(u)$. Since~$|N(I)|\leq|C(u)|$ we find that~$N(I)$ is of at most the same size as~$C(u)$ but with a smaller connected component $I$ for~$u$, contradicting the fact that~$C(u)$ is the unique minimum closest set that separates~$u$ from all other vertices of demand at least~$\phi(u)$. This completes the proof of the lemma.\qed
\end{longproof}

Now, we can give the promised bound on the number of nonzero demand vertices. %

\begin{lemma}\label{lemma:bounddemand:mk2}
Let $(G,\phi,k)$ be an instance of \VCon that is reduced according to Rule~\ref{rule:reducedemand} and let \opt denote the minimum size of feasible solutions $S\subseteq V$ for this instance. Then there are at most $d^2\opt$ nonzero demand vertices in $G$. 
\end{lemma}

\begin{longproof}
For analysis, let $S\subseteq V$ denote any feasible solution of size $\opt$, i.e., such that every $v$ with $\phi(v)\geq 1$ has $v\in S$ or there are $\phi(v)$ vertex-disjoint paths from $v$ to $S$ that overlap only in $v$. We will prove that for all $r\in\{1,\ldots,d\}$ there are at most $2r-1$ vertices of demand $r$ in $G$ (according to $\phi$). 
Fix some $r\in\{1,\ldots,d\}$ and let $D_r$ denote the set of vertices with demand exactly $r$. For each $v\in D_r$ the solution $S$ must contain at least one vertex of $R(v)$ since $C(v)=N(R(v))$ has size at most $r-1$. (Otherwise, $C(v)$ would be a $v,S$-separator of size less than $r$.) Fix some vertex $p\in S$ and let $v_1,\ldots,v_\ell$ denote all vertices of demand $r$ that have $p\in R(v_i)$. We will prove that $\ell\leq(2r-1)$ and $|D_r|\leq \opt(2r-1)$.

At most $r-1$ vertices are contained in $C(v_i)$ for every $i\in\{1,\ldots,\ell\}$. Thus, on the one hand the total size $\sum|C(v_i)|$ of these sets is at most $(r-1)\ell$. On the other hand, for every pair $v_i,v_j$ with $1\leq i<j\leq \ell$ we know that $v_i\in C(v_j)$ or $v_j\in C(v_i)$ by Lemma~\ref{lemma:cuts}, since $R(v_i)\cap R(v_j)\supseteq\{p\}\neq\emptyset$. Thus, every pair contributes at least a value of one to the total size of the sets $C(v_i)$. (To see that different pairs have disjoint contributions note that, e.g., $v_i\in C(v_j)$ uniquely defines pair $v_i,v_j$.) We get the following inequality:
\begin{align*}
(r-1)\ell\geq \sum_{i=1}^\ell |C(v_i)|\geq \binom{\ell}{2}\text{.}
\end{align*}
Thus, $\frac{1}{2}\ell(\ell-1)\leq (r-1)\ell$, implying that $\ell\leq 2r-1$. Since there are exactly \opt choices for $p\in S$ and every set $R(v)$ for $v\in D_r$ must be intersected by $S$, we get an upper bound of
\[
\opt\cdot \ell\leq \opt(2r-1) 
\]
for the size of $D_r$. If we sum this over all choices of $r\in\{1,\ldots,d\}$ we get an upper bound of
\[
\sum_{r=1}^d \opt(2r-1)=\opt\sum_{r=1}^d 2r-1=\opt\cdot d^2
\]
for the number of vertices with nonzero demand. This completes the proof.\qed
\end{longproof}

Lemma~\ref{lemma:bounddemand:mk2} directly implies reduction rules for \VdConk and \VConk: For the former, if there are more than $d^2k$ vertices then $\opt$ must exceed $k$ and we can safely reject the instance. For the latter, there can be at most $k$ vertices of demand greater than $k$ since those must be in the solution. Additionally, if $\opt\leq k$ then there are at most $d^2\opt\leq k^3$ vertices of demand at most $d=k$, for a total of $k^3+k$.

We only spell out the rule for \VdConk because it will be used in our kernelization. The bound of $k^3+k$ for \VConk will be used for the FPT-algorithm in Section~\ref{section:fpt}.

\begin{redrule}\label{rule:bounddemand}
Let~$(G,\phi,k)$ be reduced according to Rule~\ref{rule:reducedemand}, with $\phi\colon V(G)\to\{0,\ldots,d\}$. If there are more than~$d^2k$ vertices of nonzero demand return \no. %
\end{redrule}

\section{Approximation algorithm}\label{section:approximation}

In this section we discuss the approximability of \VdCon. We know from Lemma~\ref{lemma:bounddemand:mk2} that the number of vertices with nonzero demand is at most $d^2\opt$ where $\opt$ denotes the minimum size solution for the instance in question. This directly implies a factor $d^2$ approximation because taking all nonzero demand vertices constitutes a feasible solution. We now show that we can improve on this and develop a factor $d$ approximation for \VdCon.

The approximation algorithm will work as follows: We maintain a partial solution $S_0\subseteq V$, which is initially empty. In each round, we will add at most $d$ vertices to $S_0$ and show that this always brings us at least one step closer to a solution, i.e., the number of additional vertices that need to be added to $S_0$ shrinks by at least one.
To achieve this, we need to update Rule~\ref{rule:reducedemand} to take the partial solution~$S_0$ into account.

\begin{redrule}\label{rule:reducedemand:partialsolution}
Let $(G,\phi,k)$ be an instance of \VCon and let $S_0\subseteq V(G)$.
If there is a vertex~$v$ with non-zero demand and a vertex set~$W$ not containing~$v$ such that each vertex in~$W$ has demand at least~$\phi(v)$ and $v$~has at least $\phi(v)$ vertex-disjoint paths to $S_0\cup W$, then set the demand of $v$ to zero.  Similarly, if $v\in S_0$ then also set its demand to zero.
\end{redrule}

Intuitively, vertices in $S_0$ get the same status as vertices with demand at least $\phi(v)$ for applying the reduction argument. The proof of correctness now has to take into account that we seek a solution that includes $S_0$ but the argument stays essentially the same.

\begin{lemma}\label{lemma:reducedemand:partialsolution:safe}
Let $(G,\phi,k)$ be an instance of \VCon, let $S_0\subseteq V(G)$, and let $(G,\phi',k)$ be the instance obtained via a single application of Rule~\ref{rule:reducedemand}. For every $S\subseteq V(G)$ it holds that $S\cup S_0$ is a solution for $(G,\phi,k)$ if and only if $S\cup S_0$ is a solution for $(G,\phi',k)$.
\end{lemma}

\begin{longproof}
Let $v$ denote the vertex whose demand was set to zero by the reduction rule and define $r:=\phi(v)$. Clearly, $\phi(u)=\phi'(u)$ for all vertices $u\in V(G)\setminus\{v\}$, and $\phi'(v)=0$. It suffices to show that if $S\cup S_0$ fulfills demands according to~$\phi'$ then $S\cup S_0$ fulfills also demands according to~$\phi$ since $\phi(u)\geq \phi'(u)$ for all $u\in V(G)$. This in turn comes down to proving that $S_0$ fulfills the demand of $r$ at $v$ assuming that it fulfills demands according to~$\phi'$.
If $v\in S\cup S_0$ then the demand at $v$ is trivially fulfilled; this is holds for all $S$ when $v\in S_0$. Thus, we assume henceforth that $v\notin S\cup S_0$. (In particular, the case that $v\in S_0$ is done.)

Let $w_1,\ldots,w_r$ denote the $r$ vertices to which we have assumed disjoint paths from~$v$ to exist. Each of those vertices is in $S_0$ or it has demand at least $r$. Existence of these paths is required to apply the reduction rule when $v\notin S_0$.

Assume for the sake of contradiction that $S\cup S_0$ does not satisfy the demand of $r$ at $v$ (recall that $v\notin S\cup S_0$, by assumption). That is, there are fewer than $r$~vertex-disjoint paths from $v$ to $S\cup S_0$ that overlap only in $v$. It follows directly that there is a $v,S\cup S_0$-separator $C$ of size at most $r-1$. (Recall that $C$ may contain vertices of $S\cup S_0$ but not the vertex $v$.) Let~$R$ denote the connected component of~$v$ in~$G-C$. Then the following holds for each vertex $w_i\in\{w_1,\ldots,w_r\}$:
\begin{enumerate}
 \item If $w_i\in S\cup S_0$ then $w_i\notin R$: Otherwise, we would have $S\cap R\supseteq\{w_i\}\neq\emptyset$ contradicting the fact that $v$ can reach no vertex of $S\cup S_0$ in $G-C$.
 \item If $w_i\notin S\cup S_0$ then $w_i\notin R$: Since $S\cup S_0$ fulfills demands according to $\phi'$, and $w_i\notin S\cup S_0$ there must be at least $r$ vertex-disjoint paths from $w_i$ to $S\cup S_0$ that overlap only in $w_i$. However, if $w_i\in R$, then the set $C$ is also a $w_i,S\cup S_0$-separator; a contradiction since $C$ has size less than $r$.
\end{enumerate}
Thus, no vertex from~$w_1,\ldots,w_r$ is contained in~$R$. This, however, implies that~$C$ separates~$v$ from~$\{w_1,\ldots,w_r\}$, contradicting the fact that there are~$r$ vertex-disjoint paths from~$v$ to~$\{w_1,\ldots,w_r\}$ that overlap only in $v$. It follows that no such $v,S\cup S_0$-separator~$C$ can exist, and, hence, that there are at least~$r=\phi(v)$ vertex-disjoint paths from~$v$ to~$S\cup S_0$, as claimed. Thus, $S\cup S_0$ fulfills the demand of $r$ at $v$ and hence all demand according to $\phi$. (Recall that the converse is trivial since $\phi(u)\geq \phi'(u)$ for all vertices $u\in V(G)$.)\qed
\end{longproof}

It follows, that we can safely apply Rule~\ref{rule:reducedemand:partialsolution}, as a variant of Rule~\ref{rule:reducedemand}, in the presence of a partial solution $S_0$. It is easy to see that also Rule~\ref{rule:reducedemand:partialsolution} can be applied exhaustively in polynomial time because testing for any vertex $v$ is a single two-way min-cut computation and each successful application lowers the number of nonzero demand vertices by one.

We now describe our approximation algorithm. The algorithm maintains an instance $(G,\phi)$, a set $S_0\subseteq V(G)$, and an integer $\ell\in\N$. Given an instance $(G,\phi)$ the algorithm proceeds in rounds to build $S_0$, which will eventually be a complete (approximate) solution. We start with $S_0=\emptyset$ and $\ell=0$. In any single round, for given $(G,\phi)$, set $S_0$, and integer $\ell$ the algorithm proceeds as follows:

\begin{enumerate}
 \item Exhaustively apply Rule~\ref{rule:reducedemand:partialsolution} to $(G,\phi)$ and $S_0$, possibly changing $\phi$.\label{step:reduction}
 \item If $S_0$ satisfies all demands of $(G,\phi)$ then return $S_0$ as a solution (and stop).\label{step:done}
 \item Otherwise, pick a vertex $v\in V(G)$ \emph{of minimum nonzero demand}. Because we have exhaustively applied Rule~\ref{rule:reducedemand:partialsolution} there must be a set $C$ of less than $\phi(v)\leq d$ vertices that separates $v$ from $S_0$ and all vertices of demand at least $\phi(v)$.
 Add $\{v\}\cup C$ to $S_0$ and increase $\ell$ by one. Note that we add at most $\phi(v)\leq d$ vertices to $S_0$ because $|C|<\phi(v)$.\label{step:main}
 \item Repeat, i.e., start over with Step~\ref{step:reduction}.
\end{enumerate}

We claim that the algorithm preserves the following invariant.
\begin{invariant}\label{invariant:one}
  There exists a set $S_1$ of at most $\opt-\ell$ vertices such that $S_0\cup S_1$ is a feasible solution for $(G,\phi)$.
\end{invariant}
We note that the function $\phi$ may be changed throughout the algorithm, due to applications of Rule~\ref{rule:reducedemand:partialsolution}. Observe that the invariant holds trivially in the beginning, as then $S_0=\emptyset$ and $\ell=0$. We now prove that each round of our algorithm preserves the invariant. 

\begin{lemma}
Each round of the algorithm above preserves Invariant~\ref{invariant:one}.
\end{lemma}

\begin{proof}
Clearly, by Lemma~\ref{lemma:reducedemand:partialsolution:safe}, the invariant is preserved in Step~\ref{step:reduction} because the sets $S$ that extend $S_0$ to a solution stay the same. If the algorithm terminates in Step~\ref{step:done} then no more changes are made to $S_0$ or $\ell$ so the invariant still holds. It remains to discuss the interesting case that Step~\ref{step:main} happens and we add $\{v\}\cup C$ to $S_0$ and increase $\ell$ by one.

For ease of discussion let $S'_0$ and $\ell'$ denote $S_0$ and $\ell$ from before Step~\ref{step:main}. Similarly, fix a set~$S'_1$ of at most $\opt-\ell'$ vertices such that $S'_0\cup S'_1$ is a feasible solution, as promised by the invariant. We will show that there is a set $S_1$ of at most $\opt-\ell=\opt-\ell'-1$ vertices that extends $S_0=S'_0\cup\{v\}\cup C$ to a feasible solution.

Let $R$ the connected component of $v$ in $G-C$ and recall that $C$ separates $v$ from all vertices in $S_0$ and all other vertices of demand at least $\phi(v)$. Because we picked $v$ with minimum nonzero demand, $C$ must in fact separate $v$ from all other nonzero demand vertices. Thus, since Rule~\ref{rule:reducedemand:partialsolution} has been applied exhaustively, in $R$ there is no vertex of $S_0$ and no other nonzero demand vertex. The former implies, because $N(R)\subseteq C$ and $|C|<\phi(v)$ that $S'_1$ must contain at least one vertex of $R$, say $p\in S'_1\cap R$. (The latter will be used in a moment.)

We set $S_1=S'_1\setminus\{p\}$, noting $|S_1|=|S'_1|-1$, and claim that $S_0\cup S_1$ is a feasible solution; this would establish that Invariant~\ref{invariant:one} holds after Step~\ref{step:main}. Let us consider an arbitrary nonzero demand vertex $w$ and check that its demand is satisfied by $S_0\cup S_1$.
\begin{itemize}
 \item If $w\in C$ then $w\in S_0\subseteq S_0\cup S_1$ and its demand is trivially satisfied.
 \item If $w\in V\setminus (R\cup C)$ then $w\in S_0\cup S_1$ if and only if $w\in S'_0\cup S'_1$ because all changes to these sets are in $R\cup C$. If $w\notin S_0\cup S_1$ then also $w\notin S'_0\cup S'_1$ and there must be $r=\phi(w)$ vertex disjoint paths from $w$ to $S'_0\cup S'_1$, say $P'_1,\ldots,P'_r$. We change the paths to end in $S_0\cup S_1$: All paths that intersect $C$ can be shortened to end in $C$. Afterwards, no paths ends in $p$ because it would have to pass $C$ first. Thus all obtained paths, say $P_1,\ldots,P_r$ go from $w$ to $S_0\cup S_1$, and they are vertex-disjoint because they are subpaths of the vertex disjoint paths $P'_1,\ldots,P'_r$ (apart from sharing $w$, of course).
 \item Finally, if $w\in R$ then we recall that $R$ contains no other nonzero demand vertices except for $v$. Thus $w=v$ and its demand is fulfilled by $v\in S_0\subseteq S_0\cup S_1$.
\end{itemize}

We find that all steps of our algorithm maintain the invariant, as claimed.\qed
\end{proof}

Now we can wrap up the section.

\begin{theorem}\label{theorem:betterapprox}
The \VdCon problem admits a polynomial-time factor $d$ approximation.
\end{theorem}

\begin{longproof}
The algorithm works as outlined previously in this section. Given an instance $(G,\phi)$ of \VdCon we start with $S_0=\emptyset$ and $\ell=0$ and run the algorithm. We recall that these values of $S_0$ and $\ell$ fulfill Invariant~\ref{invariant:one}, recalling that $\opt$ denotes the optimum value for vector connectivity sets for $(G,\phi)$. In each round, the algorithm adds at most $d$ vertices to $S_0$, increases $\ell$ by one, and always preserves the invariant. Thus, latest when $\ell=\opt$ after $\opt$ rounds, it must stop in Step~\ref{step:done} because the invariant guarantees that some set of at most $0=\opt-\ell$ further vertices gives a solution together with $S_0$, i.e., it must find that $S_0$ is itself a solution. It then outputs $S_0$, which, after at most $\opt$ rounds, has size at most $d\cdot\opt$. This proves the claimed ratio.

We had already briefly argued that Rule~\ref{rule:reducedemand:partialsolution} can be applied exhaustively in polynomial time. Similarly, finding the required cut $C$ for a vertex $v$ of minimum demand is polynomial time, and the same is true for testing whether $S_0$ satisfies all demands. Thus, we indeed have a polynomial-time algorithm that achieves a factor $d$ approximation for \VdCon.\qed
\end{longproof}

We can also derive an approximation algorithm for \VCon, where there is no fixed upper bound on the maximum demand. To this end, we can rerun the previous algorithm for all ``guesses'' of $\opt_0\in\{1,\ldots,n\}$. In each run, we start with $S_0$ containing all vertices of demand greater than the guessed value $\opt_0$, since those must be contained in every solution of total size at most $\opt_0$. Then the maximum demand is $d=\opt_0$ and we get a $d$-approximate set of vertices to add to $S_0$ to get a feasible solution. When $\opt_0=\opt$, then $\opt$ must also include the same set $S_0$ and for the remaining $\opt-|S_0|\leq\opt$ vertices we have a $d$-approximate extension; we get a solution of total size at most $\opt^2$.

\begin{corollary}
The \VCon problem admits a polynomial-time approximation algorithm that returns a solution of size at most $\opt^2$, where $\opt$ denotes the optimum solution size for the input.
\end{corollary}

\section{FPT algorithm for Vector Connectivity($k$)}\label{section:fpt}

In this section we present a randomized FPT-algorithm for \VConk. (We recall that Lokshtanov~\cite{Lokshtanov14} announced this to be FPT.) Recall that the reduction rules in Section~\ref{section:reductionrule} also allow us to reduce the number of nonzero demand vertices to at most $k^3+k$ (or safely reject). Based on this we are able to give a randomized algorithm building on a randomized FPT algorithm of Marx~\cite{Marx09} for intersection of linear matroids. (The randomization comes from the need to construct a representation for the required matroids.) Concretely, this permits us to search for an independent set of size $k$ in $k^3+k$ linear matroids over the same ground set, where the independent set corresponds to the desired solution and each single matroid ensures that one demand vertex is satisfied. (A matroid is linear if it contains as independent sets precisely the sets of linearly independent columns of a fixed matrix. For more information see, e.g., Oxley's book~\cite{Oxley-matroidtheory}.)

\begin{theorem}\label{theorem:vconk:fpt}%
\VConk is randomized fixed-parameter tractable. The error probability is exponentially small in the input size and the error is limited to false negatives.
\end{theorem}

\begin{proof}[Sketch]
We sketch an alternative proof for fixed-parameter tractability of \VConk. W.l.o.g., input instances $(G,\phi,k)$ are already reduced to at most $k^3+k$ nonzero demand vertices (else apply the reduction rules); let $D=\{v\in V(G)\mid \phi(v)\geq 1\}$. Clearly, if the instance is \yes then there exist also solutions of size \emph{exactly $k$} (barring $|V(G)|<k$ which would be trivial).

\emph{Algorithm.}
As a first step, we guess the intersection of a solution $S^*$ of size $k$ with the set $D$; there are at most $(k^3+k)^k$ choices for $S_0=D\cap S^*$. Note that all vertices of demand exceeding $k$ must be contained in $S_0$ for $S^*$ to be a solution (we ignore $S_0$ if this is not true). 

For each $v\in D\setminus S_0$, we construct a matroid $M_v$ over $V'=V\setminus D$ %
as follows.
\begin{itemize}
 \item Build a graph $G_v$ by first adding to $G$ additional $c-1$ copies of $v$, called $v_2,\ldots,v_c$, where $c=\phi(v)$, and use $v_1:=v$ for convenience. Second, add $r=k-c$ universal vertices $w_1,\ldots,w_r$ (i.e., neighborhood $V\cup\{v_2,\ldots,v_c\}$).
 \item Let $M'_v$ denote the gammoid on $G_v$ with source set $T=\{v_1,\ldots,v_c,w_1,\ldots,w_r\}$ and ground set $V'\cup S_0$. Recall that the independent sets of a gammoid are exactly those subsets $I$ of the ground set that have $|I|$ vertex-disjoint paths from the sources to $I$. It is well known that gammoids are linear matroids and that a representation over a sufficiently large field can be found in randomized polynomial time (cf. Marx~\cite{Marx09}).
 \item Create $M_v$ from $M'_v$ by contracting $S_0$, making its ground set $V'$. If $S_0$ is independent in $M'_v$ then any $I$ is independent in $M_v$ if and only if $S_0\cup I$ is independent in $M'_v$.
\end{itemize}
Use Marx' algorithm \cite{Marx09} to search for a set $I^*$ of size $k-|S_0|$ that is independent in each matroid $M_v$ for $v\in D\setminus S_0$. If a set $I^*$ is found then test whether $S_0\cup I^*$ is a vector connectivity set for $(G,\phi,k)$ by appropriate polynomial-time flow computations. If yes then return the solution $S_0\cup I^*$. Otherwise, if $S_0\cup I^*$ is not a solution or if no set $I^*$ was found then try the next set $S_0$, or answer \no if no set $S_0$ is left to check.

It remains to prove that the algorithm is correct and to (briefly) consider the runtime.

\emph{Correctness.} Clearly, if $(G,\phi,k)$ is \no then the algorithm will always answer \no as all possible solutions $S_0\cup I^*$ are tested for feasibility.

Assume now that $(G,\phi,k)$ is \yes, let $S^*$ a solution of size $k$, and let $S_0=D\cap S^*$. Note that $S^*\subseteq V'\cup S_0$. Pick any $v\in D\setminus S_0$. It follows that there are $c=\phi(v)$ paths from $v$ to $S^*$ in $G$ that are vertex-disjoint except for $v$. Thus, in $G_v$ we get $c$ (fully) vertex-disjoint paths from $\{v_1,\ldots,v_c\}$ to $S^*$, by giving each path a private copy of $v$. We get additional $r=k-c$ paths from $\{w_1,\ldots,w_r\}$ to the remaining vertices of $S^*$ since $S^*\subseteq V'\cup S_0\subseteq N(w_i)$. Thus, the set $S^*$ is independent in each gammoid $M'_v$.
Therefore, in each $M'_v$ also $S_0\subseteq S^*$ is independent. This implies that in $M_v$ (obtained by contraction of $S_0$) the set $S^*\setminus S_0$ is independent and has size $k-|S_0|$. Moreover, any $I$ is independent in $M_v$ if and only if $I\cup S_0$ is independent in $M'_v$. It follows, from the former statement, that Marx' algorithm will find \emph{some} set $I$ of size $k-|S_0|$ that is independent in all matroids $M_v$ for $v\in D\setminus S_0$. 

We claim that $I\cup S_0$ is a vector connectivity set for $(G,\phi,k)$.
Let $v\in D\setminus S_0$. We know that $I$ is independent in $M_v$ and, thus, $S:=I\cup S_0$ is independent in $M'_v$. Thus, in $G_v$ there are $|S|=k$ paths from $T$ to $S$. This entails $c=\phi(v)$ vertex-disjoint paths from $\{v_1,\ldots,v_c\}$ to $S$ that each contain no further vertex of $T$ since $|T|=k$. By construction of $G_v$, we directly get $\phi(v)$ paths from $v$ to $S$ in $G$ that are vertex-disjoint except for overlap in $v$. Thus, $S$ satisfies the demand of any $v\in D\setminus S_0$. Since $S\supseteq S_0$, we see that $S$ satisfies all demands. Thus, the algorithm returns a feasible solution, as required.

\emph{Runtime.} Marx' algorithm for finding a set of size~$k'$ that is independent in $\ell$ matroids has \FPT running time with respect to $k'+\ell$. We have $k' \leq k$ and $\ell\leq |D|\leq k^3+k$ in all iterations of the algorithm and there are at most $(k^3+k)^k$ iterations. This gives a total time that is \FPT with respect to $k$, completing the proof sketch.\qed
\end{proof}

\section{Vertex-linear kernelization for constant demand}\label{section:kernelization}

In this section we prove a vertex-linear kernelization for \VdConk, i.e., with $d$ a problem-specific constant and $k$ the parameter. We recall the problem definition.

\parproblemdef{\VdConk}{A graph $G=(V,E)$, a function~$\phi\colon V\to\{0,\ldots,d\}$, and an integer $k\in\N$.}{$k$.}{Is there a set $S$ of at most $k$ vertices such that each vertex $v\in V\setminus S$ has $\phi(v)$ vertex-disjoint paths with endpoints in $S$?}

The starting point for our kernelization are Reduction Rules~\ref{rule:reducedemand} and~\ref{rule:bounddemand}, and a result of Cicalese et al.~\cite{CicaleseMR14} that relates vector connectivity sets for $(G,\phi)$ to hitting sets for a family of connected subgraphs of $G$: Intuitively, if the neighborhood of $X \subseteq V$ in $G$ is smaller than the largest demand of any $v\in X$, then every solution must select at least one vertex in $X$ to satisfy $v$ (by Menger's Theorem).
We begin by introducing notation for such a set family but additionally restrict it to (inclusionwise) \emph{minimal} sets $X$ where the demand of some vertex in $X$ exceeds $|N(X)|$. We then state the result of Cicalese et al.~\cite{CicaleseMR14} using our notation.

\begin{definition}[$\X(G,\phi)$]\label{definition:x}
Let $G=(V,E)$ and let $\phi\colon V\to\N$. The \emph{family $\X(G,\phi)$} contains all \emph{minimal} sets $X\subseteq V$ such that
 \begin{enumerate}
  \item $G[X]$ is connected and
  \item there is a vertex $v\in X$ with $\phi(v)>|N(X)|$.
 \end{enumerate}
\end{definition}

Using this notation, the result of Cicalese et al.~\cite{CicaleseMR14} is as follows.

\begin{proposition}[adapted from Cicalese et al.~{\cite[Proposition 1]{CicaleseMR14}}]\label{prop:HSequiv}
Let $G=(V,E)$, let $\phi\colon V\to\N$, and let $\X:=\X(G,\phi)$. Then every set $S\subseteq V$ is a vector connectivity set for $(G,\phi)$ if and only if it is a hitting set for $\X$, i.e., it has a nonempty intersection with each $X\in\X$.
\end{proposition}

\begin{proof}
Cicalese et al.~\cite{CicaleseMR14} proved Proposition~\ref{prop:HSequiv} without the minimality restriction, allowing for a larger family, say $\X^+\supseteq\X$. Clearly, hitting sets for $\X^+$ are also hitting sets for $\X$. Conversely, since $\X^+\setminus\X$ contains only supersets of sets in $\X$, hitting sets for $\X$ are also hitting sets for $\X^+$.\qed
\end{proof}

Note that for the general case of \VCon with unrestricted demands the size of $\X(G,\phi)$ can be exponential in $|V(G)|$; for \VdCon there is a straightforward bound of $|\X|=\Oh(|V(G)|^d)$ since $|N(X)|\leq d-1$. However, even for \VdCon, the sets $X\in\X$ are not necessarily small and, thus, we will not take a hitting set approach for the kernelization; in fact, we will not even materialize the set $\X$ but use it only for analysis.

We will leverage Reduction Rules~\ref{rule:reducedemand} and~\ref{rule:bounddemand} throughout this section. Hence, we will, sometimes tacitly, assume that all instances of \VCon{} are reduced with respect to these rules.

As a first step, we prove that instances $(G,\phi,k)$ of \VdConk that are reduced under Rule~\ref{rule:reducedemand} have the property that every set $X\in\X(G,\phi)$ contains at most $d^3$ vertices with nonzero demand. For ease of presentation we define $D(G,\phi):=\{v\in V(G)\mid \phi(v)\geq 1\}$, and use the shorthand $D=D(G,\phi)$ whenever $G$ and $\phi$ are clear from context.

\begin{lemma}\label{lemma:bounddemandinx}
  For all $X \in \X$ we have~$|X \cap D| \leq (d - 1)d^2 \leq d^3$.
\end{lemma}
\begin{longproof}
  Recall the definition of~$C(v)$ as the unique closest minimum $v,D'(v)$-separator, where $D'(v)=\{u\in V\setminus\{v\}\mid \phi(u)\geq\phi(v)\}$, and the definition of~$R(v)$ as the connected component of $v$ in $G-C(v)$, from Section~\ref{section:reductionrule}.  
  Fix some $r\in\{1,\ldots,d\}$, define $D_r=\{v\in D\mid \phi(v)=r\}$, and consider the relation of $R(v)$ with $X$ for $v\in D_r\cap X$. Note that $D_r\setminus\{v\}\subseteq D'(v)$.
  
  If $R(v)\subsetneq X$, then this would contradict minimality of $X$: By reducedness under Rule~\ref{rule:reducedemand}, we have $|C(v)|<\phi(v)$ or else the rule would apply to $v$. But then $R(v)$ with $N(R(v))=C(v)$ fulfills the conditions for being in $\X$, except possibly for minimality. This would prevent $X\supsetneq R(v)$ from being included in $\X$. 
  Else, if $R(v)=X$, then no further vertex of $D_r$ is in $X$, since $C(v)$ is also a $v,D_r\setminus\{v\}$-separator as $D_r\setminus\{v\}\subseteq D'(v)$. In this case we get, $|X\cap D_r|=1$.
  
  In the remaining case, there is no $v\in D_r\cap X$ with $R(v)\subseteq X$. It follows that for all $v\in D_r\cap X$ we have $R(v)\cap X\neq\emptyset$ but $R(v)\nsubseteq X$. This implies $R(v)\cap N(X)\neq \emptyset$ since both $G[X]$ and $G[R(v)]$ are connected. We will use this fact to bound the number of vertices with demand $r$ in $X$.  
  Let $w\in N(X)$ and let $W\subseteq D_r\cap X$ contain those vertices $v$ of demand $r$ whose set $R(v)$ contains $w$; each $R(v)$ must contain at least one vertex in $N(X)$. Thus, for any two vertices $u,v\in W$ we find that their sets $R(u)$ and $R(v)$ have a nonempty intersection, since they share at least $w$. We can now repeat the same analysis as used in the proof of Lemma~\ref{lemma:bounddemand:mk2} to get that $|W|\leq 2r-1$. Over all choices of $w$ we get an upper bound of $(d-1)(2r-1)$ vertices of demand $r$ in $X$. 
  
  We showed that for each choice of~$r \in \{1, \ldots, d\}$ we have at most $(d-1)(2r-1)$ vertices of demand $r$ in $X$. Summing over all $r\in\{1,\ldots,d\}$ this yields an upper bound of $(d-1)d^2\leq d^3$ for $|X\cap D|$, as claimed.\qed
\end{longproof}

To arrive at our kernelization we will later establish a reduction rule that shrinks connected subgraphs with small boundary and bounded number of demand vertices to constant size. This is akin to blackbox protrusion-based reduction rules, especially as in \cite{FominLST13}, but we give an explicit algorithm that comes down to elementary two-way flow computations. To get an explicit (linear) bound for the number of subproblems, we introduce a new family $\Y$ with larger but (as we will see) fewer sets, and apply the reduction process to graphs $G[Y]$ with $Y\in\Y$ instead. Alternatively, as pointed out in the introduction, one may use a result of Bollobas for bounding the number of sets in $\X$ once they are small; a caveat is that this bound would depend on the final size of sets in $\X$, whereas we have a direct and explicit bound for $|\Y|$.

\begin{definition}[$\Y(G,\phi,d)$]\label{definition:y}
Let $G=(V,E)$, let $d\in\N$, and let $\phi\colon V\to\{0,\ldots,d\}$. The family $\Y(G,\phi,d)$ contains all sets $Y\subseteq V$ with
 \begin{enumerate}
  \item $G[Y]$ is connected,
  \item $|Y\cap D|\leq \demandiny$, i.e., $Y$ contains at most $\demandiny$ vertices $v$ with nonzero demand $\phi(v)$,
  \item $|N(Y)|\leq d$, i.e., $Y$ has at most $d$ neighbors, and
  \item there is a vertex $v\in Y\cap D$, i.e., $\phi(v)\geq 1$, such that $N(Y)$ is the unique closest minimum $v,D\setminus Y$-separator.  
 \end{enumerate}
\end{definition}

For $(G,\phi,d)$, we relate $\X=\X(G,\phi)$ and $\Y=\Y(G,\phi,d)$ by proving that every set $X\in \X$ is contained in at least one $Y\in\Y$. Intuitively, this proves that all ``interesting'' parts of the instance are contained in subgraphs $G[Y]$.

\begin{lemma}\label{lemma:xiny}
Let $G=(V,E)$ a graph, $d\in\N$, and $\phi\colon V\to\{0,\ldots,d\}$. Let $\X:=\X(G,\phi)$ and $\Y:=\Y(G,\phi,d)$. Then for all $X\in\X$ there exists $Y\in\Y$ with $X\subseteq Y$.
\end{lemma}

\begin{longproof}
Let $X\in \X$ and pick $v\in X$ with $\phi(v)>|N(X)|$. Let $D_1= D\setminus X$, i.e., those nonzero demand vertices that are not in $X$. Now, let $Z\subseteq V\setminus \{v\}$ the minimum $v,D_1$-separator that is closest to $v$, and let $Y$ be the connected component of $v$ in $G-Z$; thus $Z=N(Y)$. We claim that $X\subseteq Y$ and $Y \in \Y$.

First, assume for contradiction that $X\nsubseteq Y$. We use submodularity of $f\colon 2^V\to \N\colon U\mapsto |N(U)|$, which implies
\begin{align}
f(X)+f(Y)\geq f(X\cap Y) + f(X\cup Y).\label{eq:ineq3}
\end{align}
Note that $D_1\cap (X\cup Y)=(D_1\cap X)\cup(D_1\cap Y)=\emptyset$ and that $v\in X\cup Y$. It follows that $N(X\cup Y)$ is also a $v,D_1$-separator and, using that $Z$ is minimum, we get that $f(X\cup Y)=|N(X\cup Y)|\geq |Z|=|N(Y)|=f(Y)$. Plugging this into \eref{eq:ineq3} we obtain $f(X)\geq f(X\cap Y)$. Note that $v\in X\cap Y$ and $D_1\cap (X\cap Y)=\emptyset$, implying that $N(X\cap Y)$ is also a $v,D_1$-separator of size at most $|N(X)|=f(X)$. From $X\nsubseteq Y$ we get $X\cap Y\subsetneq X$. But then $X\cap Y$ is a proper subset of $X$ containing $v$ and having $|N(X\cap Y)|\leq |N(X)|<\phi(v)$. It follows that the connected component of $v$ in $G[X\cap Y]$ also has neighborhood size (in $G$) less than $\phi(v)$. This contradicts the fact that $X$ is a minimal set fulfilling the properties of Definition~\ref{definition:x}, which is required for $X\in\X$. We conclude that, indeed, $X\subseteq Y$.

Second, let us check that $Y$ fulfills the requirements for $Y\in \Y=\Y(G,\phi,d)$ (as in Definition~\ref{definition:y}). Note that $Z$ separates $v$ from all nonzero demand vertices that are not in $X$, since $D_1=D\setminus X$. Thus, every nonzero demand vertex in $Y$ is also contained in $X$, i.e., $D\cap Y\subseteq D\cap X$, which bounds their number by $d^3$ using Lemma~\ref{lemma:bounddemandinx}. It also follows that $D\setminus X=D\setminus Y$, since $X\subseteq Y$ implies $D\cap X\subseteq D\cap Y$. Furthermore, $G[Y]$ is connected and $N(Y)$ is a minimum $v,D\setminus Y$-separator that is closest to $v$; note that $D_1=D\setminus X=D\setminus Y$. Finally, since $N(X)$ is also a $v,D_1$-separator and $N(Y)$ is minimum, we conclude that $|N(Y)|\leq |N(X)|<\phi(v)\leq d$. Thus, indeed, $Y\in \Y$ as claimed.\qed
\end{longproof}

We prove that the number of sets $Y\in\Y$ is linear in $k$ for every fixed $d$. Thus, by later shrinking the size of sets in $\Y$ to some constant we get $\Oh(k)$ vertices in total over sets $Y\in \Y$.

\begin{lemma}\label{lemma:boundy}
Let $(G,\phi,k)$ an instance of $\VdConk$ with $\phi\colon V(G)\to\{0,\ldots,d\}$ and let $\Y:=\Y(G,\phi,d)$. Then $|\Y|\leq d^2k\cdot 2^{d^3+d}$.
\end{lemma}

\begin{longproof}
We prove the lemma by giving a branching process that enumerates all sets $Y\in \Y$ within the leaves of a branching tree and by showing that the tree has at most $d^2k\cdot 2^{d^3+d}$ leaves in which sets $Y$ are found. Given $G=(V,E)$, $\phi\colon V\to\{0,\ldots,d\}$, and $k\in\N$, the process works as follows. (Recall $D=\{v\in V\mid \phi(v)\geq 1\}$.)
\begin{enumerate}
 \item As a first step, \emph{branch} on choosing one vertex $v\in D$. Recall that an instance reduced with respect to Rule~\ref{rule:bounddemand} has $|D|\leq d^2k$. Hence, this branching step incurs a factor of $|D|\leq d^2k$ to the number of leaves and creates one node for each choice.
 \item Maintain disjoint sets $D_0,D_1\subseteq D$, starting with $D_0:=\{v\}$ and $D_1:=\emptyset$, and the minimum~$v,D_1$-separator~$Z$ that is closest to $v$; initially~$Z=\emptyset$. Throughout, use~$Y$ to refer to the connected component of $v$ in $G-Z$.
 \item All further \emph{branchings} work as follows: Pick an arbitrary vertex $p\in D\setminus(D_0\cup D_1)$ that is reachable from $v$ in $G-Z$. \emph{Branch} on either adding this vertex to $D_0$ or to $D_1$, creating a child node for each choice. In the branch where $p$ is added to $D_1$ update the minimum $v,D_1$-separator $Z$ closest to $v$ and update the connected component $Y$ of $v$ in $G-Z$.
 \item Terminate a branch if any of the following three \emph{termination conditions} occurs.
 \begin{enumerate}
  \item The size of $D_0$ exceeds $d^3$.
  \item The size of $Z$ exceeds $d$.
  \item No vertex of $D\setminus(D_0\cup D_1)$ is reachable from $v$ in $G-Z$.
 \end{enumerate}
\end{enumerate} 

We will now analyze this process. First, we show that every set of $\Y$ occurs as the set $Y$ in some leaf node of the process, i.e., a node to which a termination condition applies. Second, we show that the number of such leaf nodes is bounded by $|D|\cdot 2^{d^3+d}\leq d^2k\cdot 2^{d^3+d}$.

\emph{\textbf{Each set of $\Y$ occurs in some leaf.}}
We show that \emph{every} set $Y^*\in\Y$ is found as the set $Y$ of at least one leaf of the branching tree. To this end, fix an arbitrary set $Y^*\in\Y$ and let $Z^*:=N_G(Y^*)$. Furthermore, let $D^*_0:=Y^*\cap D$. Let $v\in Y^*\cap D$ such that $Z^*=N(Y^*)$ is the unique minimum, closest $v,D\setminus Y^*$-separator.
In the first branching step the process can clearly pick $v$ for its choice of vertex in $D$; in this case it continues with $D_0=\{v\}$, $D_1=\emptyset$, $Z=\emptyset$, and $Y$ is the connected component of $v$ in $G$.
Consider nodes in the branching process with current sets $Y$, $D_0$, $D_1$, and $Z=N(Y)$ fulfilling the \emph{requirements} that
\begin{itemize}
 \item $Y\supseteq Y^*$ and
 \item $D_0\subseteq D^*_0$ and $D_1\cap D^*_0=\emptyset$.
\end{itemize}
Among such nodes pick one that is either a leaf or such that neither child node fulfills the requirements.
Clearly, the node with $D_0=\{v\}$, $D_1=\emptyset$, $Z=\emptyset$, and $Y$ equal to the component of $v$ in $G$ fulfills the requirements, so we can indeed always find such a node by starting at this one and following child nodes fulfilling the requirements until reaching a leaf or until both child nodes do not fulfill the requirements.

\emph{Leaf node fulfilling the requirements.}
If the chosen node is a leaf then one of the three termination conditions must hold. Clearly, we cannot have $|D_0|>d^3$ since that would imply $|D^*_0|>d^3$, violating the definition of $\Y$ and the sets therein. Similarly, we cannot have $|Z|>d$: In this regard, note that $D_1\cap D^*_0=\emptyset$ implies that $D_1\subseteq D\setminus D^*_0=D\setminus (D\cap Y^*)=D\setminus Y^*$. It follows directly that $Z^*$, which separates $v$ from $D\setminus Y^*$, also separates $v$ from $D_1$. But then the size of $Z^*$ is an upper bound for the minimum cut size for $v,D_1$-separators and $|Z|>d$ would imply $|Z^*|\geq |Z|>d$, again violating the definition of $\Y$.

Thus, in case of a leaf node the only remaining option is that no further vertex of $D\setminus (D_0\cup D_1)$ is reachable from $v$ in $G-Z=G-N(Y)$. Recall that $Z$ is the minimum closest $v,D_1$-separator. By termination $Z$ also separates $v$ from $D\setminus (D_0\cup D_1)$, making it a closest $v,D\setminus D_0$-separator. Since $D_0\subseteq D^*_0$, it follows that $D\setminus Y^*= D\setminus D^*_0\subseteq D\setminus D_0$. Thus, the size of the minimum $v,D\setminus Y^*$-separator $Z^*$ is upper-bounded by $|Z|$ since $Z$ is also a $v,D\setminus Y^*$-separator.
Recall that we have $Y \supseteq Y^*$, which implies that $Z^*=N(Y^*)$ also separates $v$ from $Z=N(Y)$. Now, because $Z$ is closest to $v$, it is the unique $v,Z$-separator of size at most $|Z|$, which implies $Z=Z^*$ since we just derived that $|Z^*|\leq |Z|$. Thus, $Y=Y^*$ since both are identified as the connected component of $v$ in $G-Z=G-Z^*$. We conclude that $Y^*$ is equal to $Y$ in the chosen node if it is a leaf.

\emph{Internal node fulfilling the requirements.}
Now, let us consider the case that the chosen node is not a leaf of the branching tree. We want to check that at least one possible branch must lead us to a child node that also fulfills our restrictions; this would contradict our choice of node that is either a leaf or such that neither child node fulfills the requirements, implying that we necessarily pick a leaf node. Thus, for any $Y^*\in\Y$ there is a leaf of the branching tree with $Y=Y^*$. For clarity, in the following discussion we will use $Y$, $D_0$, etc. for the current node and $Y'$, $D'_0$, etc. for the considered child node in the branching tree.

Since we are not in a leaf in this case, the process chooses an arbitrary vertex $p\in D\setminus (D_0\cup D_1)$ that is reachable from $v$ in $G-Z$ to branch on. If $p\in D^*_0$ then the child node corresponding to adding $p$ to $D_0$ has $D'_0=D_0\cup\{p\}\subseteq D^*_0$ and $D'_1=D_1$. Thus, $Z'=Z$ and, hence, $Y'=Y\supseteq Y^*$. Thus, the child node fulfills all requirements; a contradiction.

Otherwise, if $p\notin D^*_0$, then $p\in D\setminus D^*_0=D\setminus Y^*$. Thus, the child node corresponding to adding $p$ to $D_1$ has $D'_0=D_0\subseteq D^*_0$ and $D'_1=D_1\cup\{p\}$, implying that $D'_1\cap D^*_0=D_1\cap D^*_0=\emptyset$. For the desired contradiction it remains to prove that $Y'\supseteq Y^*$ since we assumed that neither child fulfills the requirements.

Assume that $Y^*\nsubseteq Y'$. Thus, $Y^*\cap Y'\subsetneq Y^*$. We again use submodularity of the function $f\colon 2^V\to \N\colon U\mapsto |N(U)|$ and get
\begin{align}
f(Y^*)+f(Y')\geq f(Y^*\cap Y')+f(Y^*\cup Y'). \label{ineq2}
\end{align}
Since $v\in Y^*\cap Y'\subsetneq Y^*$ and $N(Y^*)$ is closest to $v$ it follows that $f(Y^*\cap Y')=|N(Y^*\cap Y')|> |N(Y^*)|=f(Y^*)$, by Proposition~\ref{proposition:closestsets:properties}.
Plugging this into \eref{ineq2} yields $f(Y^*\cup Y')<f(Y')$; let us check that this violates the fact that $Z'=N(Y')$ is a minimum $v,D'_1$-separator: We have $v\in Y^*\cup Y'$ and 
\[
D'_1\cap (Y^*\cup Y')=\underbrace{(D'_1\cap Y^*)}_{\subseteq D}\cup \underbrace{(D'_1\cap Y')}_{=\emptyset}=(D'_1\cap Y^*)\cap D=D'_1\cap D^*_0=\emptyset.
\]
Thus, indeed, we find that $N(Y^*\cup Y')$ is a $v,D'_1$-separator and we know from $f(Y^*\cup Y')<f(Y')$ that it is smaller than the assumed minimum $v,D'_1$-separator $Z'=N(Y')$; a contradiction.
Hence, by our choice of node that fulfills the requirements and is a leaf or neither child fulfills the requirements, we must obtain a leaf with set $Y$ equal to $Y^*$.

\emph{\textbf{Number of leaves containing some $Y^*\in \Y$.}} We have seen that every set $Y^*\in\Y$ is equal to the set $Y$ of some leaf node fulfilling certain requirements. Furthermore, leaves with $|D_0|>d^3$ or $|Z|>d$ were shown not to correspond to any $Y\in \Y$. We will now analyze the number of leaf nodes with $|D_0|\leq d^3$ and $|Z|\leq d$.

Crucially, we show that each branching increases $|D_0|+|Z|$. For adding $p$ to $D_0$ this is obvious, for adding $p$ to $D_1$ we prove this next. Concretely, we prove that adding $p$ to $D_1$ increases the minimum size of $v,D_1$-separators by at least one. (This is essentially folklore but we provide a proof for completeness.) Use $D'_1=D_1\cup\{p\}$ and let $Z$ and $Z'$ denote minimum $v,D_1$- and $v,D'_1$-separators closest to $v$; use $Y$ and $Y'$ for the components of $v$ in $G-Z$ and $G-Z'$, respectively. We use again the submodular function $f\colon 2^V\to\N\colon  U\to |N_G(U)|$ and obtain
\begin{align}
f(Y)+f(Y')\geq f(Y\cup Y')+f(Y\cap Y').\label{ineq1}
\end{align}
Both $Z$ and $Z'$ separate $v$ from $D_1$. Thus,~$D_1\cap(Y\cup Y')=\emptyset$ and $N(Y\cup Y')$ is also a $v,D_1$-separator, since it creates the component $Y\cup Y'$ for $v$ in $G-N(Y\cup Y')$. Since $Z$ is a minimum~$v,D_1$-separator we must have $f(Y\cup Y')=|N(Y\cup Y')|\geq |Z|=f(Y)$. Plugging this into \eref{ineq1} yields $f(Y\cap Y')\leq f(Y')$. If $f(Y')=|Z'|\leq |Z|=f(Y)$, i.e., if adding $p$ to $D_1$ does not increase the minimum size of $v,D_1$-separators,\footnote{These values coincide with $f(Y)$ and $f(Y')$ since we chose $Z=N(Y)$ and $Z'=N(Y')$ as minimum $v,D_1$- and $v,D'_1$-separators.} then~$f(Y\cap Y')\leq |Z|$ and $N(Y\cap Y')$ is also a $v,D_1$-separator of size at most $|Z|$. However, as $p\notin Y'$, since $Z'$ separates $v$ from $D'_1=D_1\cup\{p\}$, the set $Y\cap Y'$ is a strict subset of $Y$; this is a contradiction to $Z$ being closest to $v$. As a consequence, the size of closest, minimum $v,D_1$-separators increases whenever we branch into adding a so far reachable vertex to $D_1$.

Thus, every child node has $|D'_0|+|Z'|\geq |D_0|+|Z|+1$ and, clearly, neither value decreases when branching. Thus, the process can only reach leaf nodes with $|D_0|\leq d^3$ and $|Z|\leq d$ via internal nodes where $|D_0|+|Z|\leq d^3+d$. It follows that the number of such leaf nodes is upper bounded by $|D|\cdot 2^{d^3+d}\leq d^2k\cdot 2^{d^3+d}$. (Once $|D_0|\geq d^3$ or $|Z|\geq d$ at most one child node can lead to such a leaf, since the other child violates the restriction on $|D_0|$ or $|Z|$; these branches are not counted.)
Thus, $|\Y|\leq d^2k\cdot 2^{d^3+d}$, as claimed.\qed
\end{longproof}

\paragraph{Reducing the size of sets in \Y.} In this part, we explain and prove how to reduce the size of sets $Y\in \Y$ through modifications on the graph $G$. At a high level, this will be achieved by replacing subgraphs $G[Y]$ by ``equivalent'' subgraphs of bounded size. When this is done, we know that the total number of vertices in sets $Y\in \Y$ is $\Oh(k)$.
Since this part is somewhat technical and long, let us try to illustrate it first.

Consider a set $Y\in \Y$ and its (small) neighborhood $Z:=N_G(Y)$. Think of deciding whether~$(G, \phi, k)$ is \yes as a game between two players, Alice and Bob. Alice sees only~$G[Y\cup Z]$ and wants to satisfy the demands of all vertices in~$Y$, and Bob sees only~$G - Y$ and wants to satisfy the demands of the vertices in~$V\setminus Y$. To achieve a small solution the players must cooperate and exchange information about paths between vertices in $Z$, or between $Z$ and vertices of a partial solution, that they can \emph{provide} or that they \emph{require}.

Since our goal is to simplify $G[Y]$, all notation is given using Alice's perspective. Crucially, we know that there are only constantly many nonzero demand vertices in $Y$, which can be seen to imply that the intersection of optimal solutions with $Y$ is bounded (Lemma~\ref{lemma:small-solution-in-Y} below). Thus, Alice can try all partial solutions $S_Y\subseteq Y$ of bounded size and determine what \emph{facilities} each $S_Y$ provides for Bob, and what \emph{requirements} she has on Bob to satisfy her demand vertices using $S_Y$ and further paths through $G-Y$.

Let us be slightly more concrete about facilities and requirements, before making them fully formal. If we fix some partial solution $S_Y\subseteq Y$ then Alice can offer (as facilities) to Bob to connect some subsets of $Z$ to $S_Y$ by disjoint paths in $G[Y\cup Z]$, and to, additionally, provide paths connecting certain sets of vertices in $Z$. There can be a large number of such options for each $S_Y$. Similarly, to fulfill demands in $Y$, Alice may need (as requirements) that Bob can provide paths from certain subsets of $Z$ to a solution and, additionally, paths connecting sets of vertices in $Z$. (Note that there is some asymmetry here since Bob's part is too large to fully analyze, but this will be no problem.) Fortunately, while there may be many choices for $S_Y$, and many facilities and requirements relative to a single $S_Y$, it will turn out that the overall number of things to keep track of is bounded (in $d$); this will be called the \emph{signature} of $G[Y\cup Z]$. Ultimately, we will be able to replace $G[Y\cup Z]$ by a bounded-size graph with the same signature.

We now make our approach formal. For convenience, let us introduce the following notation. %
\newcommand{\indep}[1]{\ensuremath{#1}\nobreakdash-independent\xspace}%
\newcommand{\indepc}[1]{\ensuremath{#1}\nobreakdash-independence\xspace}%
A \emph{separation} of a graph~$G$ is a tuple $(T, U)$ of two vertex subsets $T, U \subseteq V(G)$ such that $T \cup U = V(G)$ and there is no edge between $T \setminus U$ and $U \setminus T$ in $G$. The \emph{order} of a separation~$(T, U)$ is $|T \cap U|$. 
We call a set of paths to be \emph{\indep{v}} if each pair of paths is vertex-disjoint except for possibly sharing $v$ as an endpoint. For a graph $G$, a vertex~$v$, an integer~$i$, and two vertex subsets $A, B \subseteq V(G)$ we define a \emph{$(v, i, A, B)$-constrained path packing} as a set of $i + |A|$ \indep{v} paths from $A \cup \{v\}$ to $B$ in $G$. Herein we explicitely allow $v \notin V(G)$ if $i = 0$.  If $i = 0$, we simplify the notation and speak of $(A, B)$-constrained path packings instead. Note that, regardless of whether $v \in V(G)$, the paths saturate each vertex in~$A$. Furthermore, we tacitly assume that, if $A \cap B \neq \emptyset$ then the paths corresponding to $A \cap B$ in the packing are of length zero, that is, they each comprise a single vertex.

Let us now begin with the definition of signatures. In the following, let~$G$ and $\phi$ %
represent the graph and demand function of an instance of \VdConk. To decrease the necessary notation we also fix a set~$Y \subseteq V(G)$ such that~$|Y \cap D| \leq \demandiny$ and $|N(Y)| \leq d$. (In the kernelization procedure, the role of~$Y$ will be assumed by some set in~$\Y$.) We denote the neighborhood~$N(Y)$ by~$Z$. 

We will first take care of the requirements that Alice has. The facilities will be treated later. As mentioned above, a requirement comprises several sets of paths outside of~$G[Y]$ one set of which needs to be provided by Bob (and his part of the solution) in order to satisfy the demand of a vertex~$v \in D \cap Y$. To this end, we define a \emph{satisfying connector}. In the following $Y \uplus Z$ denotes the disjoint union of the sets $Y$ and~$Z$.

\begin{definition}[Satisfying connector]
  Let~$H$ be a graph on vertex set $Y \uplus Z$, let $v \in Y$, let $v$ have positive demand~$d_v$, and let $S_{Y} \subseteq Y$ be a partial solution. A tuple $(A, B, C)$ with $A, B, C\subseteq Z$, pairwise disjoint, is a \emph{satisfying connector for~$v$ with respect to~$S_Y$ in~$H$} if either~$v \in S_Y$ or there is a $(v, d_v, A, B \cup S_Y)$-constrained path packing in $H - C$. The set of all satisfying connectors for~$v$ with respect to the partial solution~$S_Y$ is denoted by~$\satcon(H, Z, d_v, S_Y, v)$.
\end{definition}

Now we can formally define the requirement of Alice. Intuitively, Bob
should provide paths that ``hit'' each set of satisfying connectors
induced by a vertex with positive demand.

\begin{definition}[Requirement]
  Let~$H$ be a graph on vertex set~$Y \uplus Z$, let~$\phi$ be a demand function on~$Y$, and let~$S_Y \subseteq Y$ be a partial solution. The \emph{requirement}~$\req(H, Z, \phi, S_Y)$ is the collection~$\{\satcon(H,Z,\phi(v),S_Y,v) \mid v \in
  D(H, \phi) \cap Y\}$.
\end{definition}

After the definition of signatures we will prove that the demand of each vertex in~$D(\phi) \cap Y$ can be met if and only if the requirement of Alice is met by a suitable path packing provided by Bob.

Note that, in order to be able to replace~$G[Y]$ by a different graph, we need to know which requirements it imposes for any relevant choice of the partial solution~$S_Y$. Since we are aiming for a polynomial-time kernelization, we also need to be able to compute them in polynomial time. For this, we first bound the size of~$S_Y$.

\begin{lemma}\label{lemma:small-solution-in-Y}
  For each vector connectivity set~$S$ of~$(G, \phi)$, there is a vector connectivity set~$S'$ such that~$|S'| \leq |S|$ and such that~$|Y \cap S'| \leq \demandiny + d$.
\end{lemma}%
\begin{longproof}
  Assume that~$|Y \cap S| \geq \demandiny + d$. Remove all vertices in~$(Y \cup Z) \cap S$ from~$S$ and add all vertices in~$D \cap Y$ as well as all vertices in~$Z$. Denote the resulting vertex set by~$S'$. Recall that $|D \cap Y| \leq \demandiny$ by definition of~$Y$. Hence, $|Y \cap S'| \leq |D \cap Y| + |N_G(Y)| \leq \demandiny + d$ and, thus, $|S'| \leq |S|$. Clearly, $S'$ satisfies all demands of vertices in~$Y$ and $Z$ since $(D\cap Y)\cup Z\subseteq S'$. For vertices $v\in D\setminus(Y\cup Z)$, note that at most $|Z|$ paths used for reaching $S$ from $v$ can traverse $Z$, and all of those can be shortened to end in $Z\subseteq S'$; all other paths avoid $Z$ and thereby $Y\cup Z$, implying that they still end in vertices of $S\setminus(Y\cup Z)=S'\setminus(Y\cup Z)$.\qed
\end{longproof}

We are almost ready to compute the requirements; crucially, we need to check whether there are suitable $(v, d, A, B)$-constrained path packings in polynomial time.

\begin{lemma}\label{lemma:path-packing-poly}
  Let $G$ be a graph, $v \in V(G)$, $d \in \mathbb{N}$, and $A, B \subseteq V(G)$. It is possible to check in polynomial time, whether there is a $(v, d, A, B)$-constrained path packing in $G$.
\end{lemma}
\begin{proof}[Sketch]
  We reduce the task to computing a maximum flow in a modified graph: First, remove each vertex in $A \cap B$ from the graph---we can assume that they represent paths of length zero in the desired path packing. Then, add $d$ copies of $v$ to $G$, each adjacent to all neighbors of~$v$. Then, attach a new vertex~$s$ to each vertex in~$A$ and the copies of $v$, and attach a new vertex $t$ to each vertex in $B$. It is not hard to check that there are $d + |A \setminus B|$ internally vertex-disjoint paths from $s$ to $t$ in the modified graph if and only if there is a $(v, d, A, B)$-constrained path packing in the original graph. 

  Checking whether there are enough internally vertex-disjoint $s$-$t$ paths can be done using a folklore reduction to maximum flow: Create a flow network by introducing two vertices $v_{\text{in}}$ and $v_{\text{out}}$ for each vertex~$v \in V(G)$, and also an arc $(v_{\text{in}}, v_{\text{out}})$ of capacity one. Then, for each edge $\{u, v\}$, add two arcs $(u_{\text{out}}, v_{\text{in}}), (v_{\text{out}}, u_{\text{in}})$ with capacity infinity. One can check that there is a flow of value $d + |A \setminus B|$ between $s_{\text{out}}$ and $t_{\text{in}}$ if and only if the desired path packing exists.\qed
\end{proof}

\begin{lemma}\label{lemma:req-poly}
  Let $Y \in \Y$, $Z = N(Y)$, and $H = G[Y \cup Z]$. The set $\{\req(H, Z, \phi|_{Y}, S_Y) \mid S_Y \subseteq Y \wedge |S_Y| \leq \demandiny + d\}$ is computable in polynomial time. Herein, $\phi|_{Y}$ denotes $\phi$ restricted to~$Y$.
\end{lemma}
\begin{longproof}
  By enumerating all $S_Y \subseteq Y$ of size at most~$\demandiny +
  d$ in $n^{\Oh(d^3)}$~time, i.e. polynomial time, the task reduces to
  computing~$\req(H, Z, \phi|_{Y}, S_Y)$ for a given~$S_Y$. To do
  this, we compute the set of satisfying connectors for each vertex
  $v$ in~$D \cap Y$. This in turn we do by simply iterating over all
  tuples~$(A, B, C)$ with $A, B, C \subseteq Z$, pairwise disjoint,
  and we check whether it is contained in~$\satcon(H, Z, \phi(v), S_Y,
  v)$. Note that~$|Z| \leq d$. Hence, there are at most~$2^{\Oh(d)}
  \in \Oh(1)$ different tuples to check. Thus, it remains to check
  whether in~$H - C$ there is a $(v, \phi(v), A, B \cup
  S_Y)$-constrained path packing. This can be done using
  Lemma~\ref{lemma:path-packing-poly}.\qed
\end{longproof}

So far we have only talked about the requirements that Alice has on Bob. Now let us come to the facilities that Alice provides. As mentioned, a facility models the sets of paths inside of~$G[Y]$ that Alice, using her part of the solution, provides to Bob so to satisfy the demand of each vertex in~$D \setminus Y$. Let us first focus on vertices in~$D \setminus (Y \cup Z)$.

\begin{definition}[Provided connector]
  Let $H$ be a graph on vertex set~$Y \uplus Z$, and let~$S_Y \subseteq Y$ be a partial solution. A tuple $(A, B, C)$ with $A, B, C \subseteq Z$, pairwise disjoint, is a \emph{provided connector of~$S_Y$ in~$H$} if there is a $(A,B \cup S_Y)$-constrained path packing in $H - C$.
\end{definition}

We prove below that all demands of vertices in $D \setminus (Y \cup Z)$ can be met if and only if Alice provides suitable path packings to Bob. We take special care of vertices in~$D \cap Z$, as they may need multiple paths into~$G[Y]$.

\begin{definition}[Provided special connector]\label{definition:providedspecialconnector:mk2}
  Let $H$ be a graph on vertex set $Y \uplus Z$ and let $S_Y\subseteq Y$. A tuple $(z,i,A,B,C)$ with $z\in Z$, $A,B,C\subseteq Z \setminus \{z\}$, pairwise disjoint, and $i\in\{0,\ldots,d\}$ is a \emph{provided special connector of $S_Y$ in $H$} if there is a $(z, i, A, B \cup S_Y)$-constrained path packing in $H - C$.
\end{definition}

We are now ready to give a formal definition of the facilities provided by Alice.

\begin{definition}[Facility]
  Let $H$ be a graph on vertex set~$Y \uplus Z$, and let~$S_Y \subseteq Y$ be a partial solution. The \emph{facility~$\fac(H, Z, S_Y)$ of $S_Y$ in $H$} is the set of all provided connectors and provided special connectors of~$S_Y$.
\end{definition}

Similarly to requirements, we need an efficient algorithm for computing the facilities; this basically follows from Lemma~\ref{lemma:path-packing-poly}.

\begin{lemma}\label{lemma:fac-poly}
  Let $Y \in \Y$, $Z = N(Y)$, and $H = G[Y \cup Z]$. The set~$\{\fac(H, Z, S_Y) \mid S_Y \subseteq Y \wedge |S_Y| \leq \demandiny + d\}$ is computable in polynomial time.
\end{lemma}
\begin{longproof}
  We first enumerate all $S_Y \subseteq Y$ with $|S_Y| \leq \demandiny + d$ in $n^{\Oh(d^3)}$~time and, for each such $S_Y$, we compute all provided connectors and provided special connectors. This is done by iterating over all possible tuples~$(A, B, C)$ and~$(z, i, A, B, C)$ (there are at most $d^2 \cdot 2^{\Oh(d)} \in \Oh(1)$ of them) and checking whether they indeed are provided (special) connectors. This is done using Lemma~\ref{lemma:path-packing-poly}.\qed
\end{longproof}

Now we can precisely define the signature of~$G[Y\cup Z]$ that we mentioned earlier.

\begin{definition}[Signature]
  Let $H$ be a graph on vertex set~$Y \uplus Z$, let $\phi$ be a demand function on~$Y$, and let~$S_Y \subseteq Y$ be a partial solution. The \emph{signature} of~$H$ is the set~$$\sig(H, Z, \phi) := \{(|S_Y|, \req(H, Z, \phi, S_Y), \fac(H, Z, S_Y))\}\text{,}$$
  where $S_Y$ ranges over all~$S_Y \subseteq Y$ such that~$|S_Y| \leq \demandiny + d$.
\end{definition}

We show below that we can safely replace~$G_Z[Y \cup Z]$ with any graph $G'[Y' \cup Z]$ that has the same signature. Let us now prove that there is a graph~$G'[Y' \cup Z]$ of constant size with the same signature.

\begin{lemma}\label{lemma:replacement-Y}
  There is a polynomial-time algorithm that receives $Y \in \Y$, $Z = N(Y)$, $G[Y \cup Z]$, and~$\phi|_{Y}$ as input and computes a graph~$G'$ on vertex set~$Y' \cup Z$ such that $G'[Z] = G[Z]$ and a demand function~$\phi' \colon Y' \to \{0, \ldots, d\}$ such that~$\sig(G[Y \cup Z], Z, \phi|_{Y}) = \sig(G', Z, \phi')$ and $|D(G', \phi')| \leq \demandiny$. Moreover, the resulting $G'$ and~$\phi'$ are encoded using at most~$f(d)$ bits and $G'$ has at most $f(d)$ vertices, for some computable function~$f$ depending only on~$d$.
\end{lemma}
\begin{longproof}
  The algorithm is as follows. First, compute $\sig(G[Y \cup Z], Z, \phi|_{Y})$ in polynomial time using Lemmas~\ref{lemma:req-poly} and~\ref{lemma:fac-poly}. Then, generate all graphs~$G'$ with $G'[Z] = G[Z]$ in the order of increasing number of vertices (break the remaining ties arbitrarily). For each graph~$G'$, iterate over all possible~$\phi' \colon V(G') \setminus Z \to \{0, \ldots, d\}$ and check whether~$\sig(G[Y \cup Z], Z, \phi|_{Y}) = \sig(G', Z, \phi')$ as well as $|D(G', \phi')| \leq \demandiny$. Clearly, this procedure terminates and finds the required tuple~$(G', \phi')$, because~$(G[Y \cup Z], \phi|_Y)$ witnesses its existence.

  Without loss of generality, we may assume~$Z = \{1, \ldots, |Z|\}$. 
  Now note that both requirements and facilities contain only set systems over~$Z$, tuples of elements of~$Z$, numbers in~$\{1,\ldots,\Oh(d^3)\}$, and the number of these entities is bounded by a function of~$d$. Hence, $\sig(G[Y \cup Z], Z, \phi|_{Y})$ can be encoded using at most~$g(|Z|) \leq g(d)$ bits, where~$g$ is some monotone ascending, computable function. Thus, since~$\sig(G[Y \cup Z], Z, \phi|_{Y})$ is the only input to the procedure that finds~$G'$ and~$\phi'$, the procedure terminates after at most $f'(g(d))$~steps for some computable function~$f'$, meaning that~$(G', \phi')$ is of size at most~$f'(g(d))$.\qed
\end{longproof}

Now we can make the notion of replacing~$G[Y]$ more precise; it involves an operation commonly referred to as glueing of graphs.

\begin{definition}[Glueing]
  Let $G_1, G_2$ be two graphs, both containing $Z$ as a subset of their vertices. \emph{Glueing $G_1$ and~$G_2$ on $Z$} results in a graph denoted by~$G_1 \oplus_Z G_2 := (V(G_1) \cup V(G_2), E(G_1) \cup E(G_2))$, where~$V(G_1)$ and~$V(G_2)$ are treated as being disjoint except for~$Z$.
\end{definition}

Below we only glue on~$Z$ for some vertex set~$Z$ defined in the context, so we will omit the index~$Z$ in the~$\oplus$ operation.

We arrive at the reduction rule aiming at reducing the size of $Y$.

\begin{redrule}\label{rule:replace-Y}
  Let~$(G, \phi, k)$ be an instance of \VdConk{} that is reduced with respect to Reduction Rules~\ref{rule:reducedemand} and~\ref{rule:bounddemand}. Let~$Y \in \Y$, where $\Y$ is as in Definition~\ref{definition:y}, and let~$Z = N_G(Y)$. Furthermore, let $(G', \phi')$ be as in Lemma~\ref{lemma:replacement-Y}. If~$|Y| > |V(G') \setminus Z|$, then replace~$G$ by~$G' \oplus (G - Y)$ and replace~$\phi$ by~$\phi|_{V(G) \setminus Y} \cup \phi'$.
\end{redrule}

Before proving that Rule~\ref{rule:replace-Y} is safe, we need a technical lemma that shows how paths from a demand vertex to a vertex-connectivity set are split over a separator.

\begin{lemma}\label{lem:path-packings-over-separator}
  Let $G$ be a graph, $(T, U)$ a separation of $G$, $Z = T \cap U$, $S \subseteq V(G)$, $v \in V(G) \setminus S$, and $d \in \mathbb{N}$. There are $d$ \indep{v} paths from~$v$ to $S$ in $G$ if and only if there is an integer~$i \in \{0, \ldots, d\}$ and a partition of $Z\setminus \{v\}$ into four vertex sets $A, B, C, D$ such that
  \begin{enumerate}
  \item if $v \in T \setminus U$ then $i = d$, and if $v \in U \setminus T$ then $i = 0$,
  \item there is a $(v, i, A, B \cup (S \setminus U))$-constrained path packing in~$G[T \setminus C]$, and %
  \item there is a $(v, d - i, B, A \cup (S \cap U))$-constrained path packing in~$G[U \setminus D]$.%
  \end{enumerate}
\end{lemma}
\begin{longproof}
  ($\Rightarrow$): Assume first that there are $d$ \indep{v} paths from $v$ to $S$ in $G$ and let $\P$ be a corresponding path packing (with overlap only in start vertex $v$). We may safely assume that paths in $\P$ have no vertices of $S$ as internal vertices; else they could be shortened. We will select $A,B,C,D\subseteq Z \setminus \{v\}$ and $i\in\{0,\ldots,d\}$ such that the path packings exist as stated in the lemma. For the purpose of getting a clear partitioning of the edges contained in $Z$, we show that one of the packings exists in $G[T \setminus C] - E(G[Z])$ and one of them in the remainder of the graph. Let us shorthand $H$ for~$G[T \setminus C] - E(G[Z])$.
  
  \newcommand{\dir}[1]{\ensuremath{\vec{#1}}\xspace}
  Consider all paths in $\P$ as being directed from $v$ towards $S$, and consider the set $\P_H$ of maximal, directed subpaths in~$H$ of paths in~$\P$ such that each path in $\P_H$ contains at least one arc. Denote by \dir{H} the directed subgraph of $H$ induced by~$\P_H$. That is, \dir{H} contains precisely the vertices and arcs also contained in the paths in~$\P_H$. We can now pick the sets $A,B,C,D\subseteq Z\setminus \{v\}$. The source vertices in~\dir{H} not equal to~$v$ form the set~$A$. Note that all source vertices except possibly $v$ are contained in~$Z$ as each vertex on a path in $\P_H$ but not in~$Z \cup \{v\}$ must have a predecessor. Similarly, sink vertices in~\dir{H} are contained in $Z \cup S$; we put those sink vertices that are contained in~$Z \setminus \{v\}$ into~$B$. Note that, as each path in $\P_H$ has length at least one, there are no vertices of in- and outdegree zero and hence $A \cap B = \emptyset$. Vertices in~$Z$ that are used by paths in $\P_H$, but that are neither sources nor sinks, are put into~$D$. Vertices of $Z\setminus \{v\}$ that are not on any path in $\P_H$ are put into~$C$. Clearly, $A, B, C, D$ is a partition of~$Z \setminus \{v\}$. Finally, we define $i = d$ if $v \in T \setminus U$, $i = 0$ if $v \in U \setminus T$, and if $v \in Z$, then $i$ is defined as the outdegree of~$v$ in~\dir{H}. The condition on $d$ and $i$ in the lemma is clearly fulfilled. We claim that $\P_H$ is the desired path packing in~$G[T \setminus C]$.

  \emph{Showing that $\P_H$ is a $(v, i, A, B \cup (S \setminus U))$-constrained path packing in $G[T \setminus C]$:} Clearly, $H$ is a subgraph of $G[T \setminus C]$ and hence $\P_H$ is contained in $G[T \setminus C]$. Observe that, since the paths in $\P_H$ are vertex-disjoint (except for $v$) and each path has length at least one, sources and sinks in \dir{H} correspond to endpoints of these paths. Combining this with our observation from above that sources and sinks in \dir{H} are in $A \cup\{v\}$ and $B \cup S$, respectively, we infer that each path in $\P_H$ starts in either~$v$ or $A$, and ends in~$B \cup (S \setminus U)$. By definition of~$A$, each vertex~$w\in A$ has a path in $\P_H$ starting in~$w$ and, furthermore, by the definition of $i$, there are $i$~paths in~$\P_H$ that start in~$v$. Hence, $\P_H$~witnesses that there are $|A| + i$~paths from~$A \cup \{v\}$ to $B \cup (S \setminus U)$ in~$H$; moreover, these paths do not touch~$C$ by definition. As the paths in~$\P$ are \indep{v}, so are the paths in~$\P_H$. Hence, $\P_H$~is a $(v, i, A, B \cup (S \setminus U))$-constrained path packing in $G[T \setminus C]$.

  \emph{Showing existence of a $(v, d - i, B, A \cup (S \cap U))$-constrained path packing in $G[U \setminus D]$:} Take the path packing~$\P$ and define a path packing $\P'$ that contains all maximal subpaths of $\P$ in~$U \setminus D$. Note that $\P'$ may contain paths of length zero. We consider also $\P'$ as a set of directed paths, each arc inheriting its direction from~$\P$. Clearly, $\P'$ is contained in $G[U \setminus D]$. We claim that $\P'$ is also $(v, d - i, B, A \cup (S \cap U))$-constrained.

  Consider the directed subgraph \dir{G'} of~$G[U]$ induced by~$\P'$. Let us find the endpoints of the paths in $\P'$. Clearly, if $v \in U \setminus D$, then $v$ is such an endpoint. For the remaining endpoints, first, consider a path of length zero, represented by a vertex~$w \neq v$. Since each path in $\P$ has length at least one and since $w \neq v$, $w$ has a predecessor~$u$ on a path in $\P$. Since $w$ represents a path of length zero, $u \in (T \setminus U) \cup D$. By definition of $D$ (since $H$ does not contain any edges in $Z$), each successor of a vertex in $D$ on a path of $\P$ is contained in $T \setminus U$. Hence, in fact $u \in T \setminus U$. The only vertices in $U$ that have neighbors in $T \setminus U$ are contained in $T \cap U = Z$. This implies that $w \in Z$ and hence $w \in Z \setminus C$. Since $\P'$ has empty intersection with $D$, we moreover have $w \notin D$. Hence, only two possibilities remain: $w \in A$ and $w \in B$. Assume that $w \in A$. Since the vertices in $A$ are sources of $\dir{H}$, this implies that there is a path starting in $w \neq v$ in $\P$, a contradiction. It follows that $w \in B$. Since $B$ represents sinks in $\dir{H}$, we moreover have $w \in S$ as, otherwise, there would be a path in $\P$ ending in a vertex not contained in $S$. Thus, as also $w \in U$, each path of length zero in $\P'$ ends in $S \cap U$ (and starts in~$B$). 

  Next, consider paths of length at least one in $\P'$. Since they are pairwise vertex-disjoint (except for $v$), their endpoints correspond to the sources and sinks in \dir{G'}. Let $w$ be a sink in \dir{G'} that is not contained in~$S$. Since $w$ is not in~$S$, it has a successor~$x$ on a path in $\P$. As above, by the definition of $D$, each predecessor of a vertex in $D$ on a path in $\P$ is contained in~$T \setminus U$. Hence, in fact $x \in T \setminus U$. The only vertices in~$U$ that have neighbors in~$T \setminus U$ are contained in $T \cap U = Z$. Hence, we have~$w \in Z$. Observe that~$w \neq v$ as, otherwise, $\P$ contains a cycle. The paths in~$\P$ are vertex-disjoint, thus, $w$ does not have any incoming arcs in~\dir{H}, meaning that it is a source in \dir{H}. This implies $w \in A$ by definition of~$A$. Thus we obtain that $\P'$ is a packing of paths, each of which ends in~$A \cup (S \cap U)$.

  It remains to prove that $\P'$ contains $|B| + d - i$ paths that start in~$\{v\} \cup B$; their \indepc{v} is implied by the fact that these paths are subpaths of paths in~$\P$. We claim that there are $d - i$ paths in $\P'$ starting in~$v$. First, if $v \notin U$ then $i = d$ by definition; hence, the claim is trivially true. If $v \in U \setminus T$ then, $i = 0$ and, clearly, each path in $\P$ that starts in $v$ induces one such path in $\P'$. Thus, the claim holds also in this case. Finally, if $v \in Z$, then $i$ is the outdegree of $v$ in $\dir{H}$. Recall that $H$ does not contain any edge in $Z$. Hence, also in the final case there are~$d - i$~paths in~$\P'$ that start in~$v$. 
  
  To find the remaining $|B|$ paths, consider a vertex~$w \in B$. Note that $w \neq v$ because $v \notin B$. By definition, $w$~is a sink in~\dir{H} and, since it is a part of a path in~$\P$ reaching $S$, it either is contained in~$S$ or has a successor on~$\P$ which is not contained in~$T \setminus C$. In the first case, $w$ represents a length-zero path starting in~$B$ in $\P'$. In the second case, $w$ is a source in \dir{G'} by the vertex-disjointness of the paths in~$\P$. Since the choice of~$w$ is arbitrary, and since the paths in $\P$ are vertex-disjoint, each vertex in $B \setminus S$ is a source in $\dir{G'}$ and hence has a path in $\P'$ starting in this vertex. Thus, overall, there are $|B| + d - i$ paths from $\{v\} \cup B$ to $A \cup (S \cap U)$ in $G[U \setminus D]$ which are \indep{v}, as required. This completes the ``if'' part of the 
  lemma.
  
  ($\Leftarrow$): Assume that there is a partition $A, B, C, D$ of $Z \setminus \{v\}$ and $i \in \mathbb{N}$ as described in the lemma and fix a $(v, i, A, B \cup (S \setminus U))$-constrained path packing $\P_T$ in $G[T \setminus C]$ and a $(v, d - i, B, A \cup (S \cap U))$-constrained path packing $\P_U$ in $G[U \setminus D]$. Consider the paths in $\P_T$ as directed from $\{v\} \cup A$ to $B \cup (S \setminus U)$ and the paths in $\P_U$ as directed from $\{v\} \cup B$ to $A \cup (S \cap U)$. Let us show that there is a packing of $d$ \indep{v} paths from $v$ to $S$ in~$G$.

  Observe that $\P_T$ and $\P_U$ may overlap only in $A \cup B \cup (\{v\} \cap Z)$, as the graphs they are contained in overlap precisely in this vertex set. Consider the directed graph induced by the union of $\P_T$ and $\P_U$. Denote by $K$ the (weakly) connected component of this graph that contains~$v$. By definition of~$\P_T$ and $\P_U$, vertex~$v$ is a source vertex. Let us first show that $v$ has outdegree $d$ in $K$. Otherwise, $v$ must have a successor~$w$ in either $A$ or $B$ in both $\P_T$ and $\P_U$. However, as the paths in $\P_T$ start in~$A$ and the paths in $\P_U$ start in~$B$ in both cases we get a contradiction. Thus, $v$ is a source with precisely~$d$ outgoing arcs in~$K$. 

  We claim that $v$ is the only source in~$K$. To see this, we first derive in- and outdegrees of all vertices other than $v$ in~$K$. Clearly, each vertex in $K - (A \cup B \cup \{v\})$ is either in~$S$---and has indegree one and outdegree zero in this case---or has in- and outdegree exactly one. We claim that each vertex in $A \cup B$ has indegree at most one in $K$. This is clear for vertices in $A$, as only $\P_U$ sends paths to $A$. Both packings $\P_T$ and $\P_U$ may send paths to a vertex $w \in B$ in the case that $w \in S$. Then, however, $w$ is in a path of length zero in $\P_U$,\footnote{Recall that in a $(v, i, A, B)$-constrained path packing, we assume each path with endpoints in $A \cap B$ to be of length zero.} implying that indeed each vertex in $A \cup B$ has indegree at most one in $K$. Now for the sake of contradiction assume that there are two sources in~$K$ and consider a path in the underlying undirected graph of~$K$ between these two sources. On this path, there is a vertex with indegree at least two; a contradiction. Thus, $v$ is the only source in~$K$. 

  It now suffices to show that each vertex in $A \cup B$ either has in- and outdegree one in $K$ or is a sink contained in~$S$. As we have derived the same for the vertices in $V(K) \setminus (A \cup B \cup \{v\})$ above, and since the sum of all indegrees equals the sum of all outdegrees in~$K$, this then implies that there are $d$~vertex disjoint paths from~$v$ to~$S$. Let thus prove that, indeed, each vertex in $A \cup B$ has either in- and outdegree one in $K$ or is a sink contained in $S$.

  We have shown above that each vertex in $A \cup B$ has indegree at most one in~$K$. Since $K$ has $v$ as its only source, each of the vertices in $A \cup B$ has also indegree at least one in $K$. Since only $\P_T$ has paths starting in~$A$ and only $\P_U$ has paths starting in~$B$, the outdegree of the vertices in $A \cup B$ is at most one in $K$. Now consider a sink $w \in V(K) \cap (A \cup B)$. Recall that $\P_T$ contains a path starting in each vertex of $A$. Since $A \cap (B \cup (S \setminus U)) = \emptyset$, each of these paths has length at least one. Hence, $w \in B$. Since also $\P_U$ has a path starting in~$w$, it must be of length zero and thus $w \in S$ because the paths in $\P_U$ end in $B \cup (S \cap U)$ and $A \cap B = \emptyset$. Thus we have shown that each vertex in $K$ is either the source~$v$ with outdegree~$d$, has in- and outdegree exactly one, or is a sink contained in $S$ and has indegree exactly one. This means that there are $d$ \indep{v} paths from $v$ to $S$ in $G$.\qed
\end{longproof}

We are ready to show that Rule~\ref{rule:replace-Y} respects \yes and \no instances.

\begin{lemma}
  Rule~\ref{rule:replace-Y} is safe.
\end{lemma}
\begin{longproof}
  We claim that an even stronger statement holds. Namely, let $G_1, G_2, \hat{G}$~be three graphs, each containing~$Z$ as a vertex subset such that $G_1[Z] = G_2[Z] = \hat{G}[Z]$. Furthermore, let~$\phi_1 \colon V(G_1)\setminus Z \to \{0, \ldots, d\}$, $\phi_2 \colon V(G_2) \setminus Z \to \{0, \ldots, d\}$, and~$\hat{\phi} \colon V(\hat{G}) \to \{0, \ldots, d\}$ be three demand functions, such that $|D(G_1, \phi_1)|, |D(G_2, \phi_2)| \leq \demandiny$ and such that $\sig(G_1, Z, \phi_1) = \sig(G_2, Z, \phi_2)$. Then $G_1 \oplus \hat{G}$ has a vector connectivity set of size~$k$ with respect to~$\phi_1 \cup \hat{\phi}$ if and only if~$G_2 \oplus \hat{G}$ has a vector connectivity set of size~$k$ with respect to~$\phi_2 \cup \hat{\phi}$. 

  To see that our claim implies the lemma, set~$G_1:=G[Y \cup Z]$, $G_2:=G'$, $\hat{G}:=G - Y$ and define the demand functions~$\phi_1, \phi_2$ and $\hat{\phi}$ accordingly. Then our claim implies that $G_1 \oplus \hat{G} = G[Y \cup Z] \oplus (G - Y) = G$ has a vector connectivity set of size~$k$ if and only if~$G_2 \oplus \hat{G} = G' \oplus (G - Y)$ has such a set. That is, it implies that Rule~\ref{rule:replace-Y} is safe.

  Let us prove the claim. Note that, by swapping the names of~$G_1$ and~$G_2$, as well as~$\phi_1$ and~$\phi_2$, it suffices to prove one direction. Assume hence that~$G_1 \oplus \hat{G}$ has a vector connectivity set~$S$ of size~$k$ with respect to~$\phi_1 \cup \hat{\phi}$. Since $|D(G_1, \phi_1)| \leq \demandiny$ and~$|N_{G_1 \oplus \hat{G}}(V(G_1) \setminus Z)| \leq |Z| \leq d$ we may apply Lemma~\ref{lemma:small-solution-in-Y} with $Y = V(G_1) \setminus Z$ and hence, we may assume that $S_1:=(V(G_1) \cap S) \setminus Z$ contains at most~$\demandiny + d$ elements. Thus, we have 
  \[
  (|S_1|, \req(G_1, Z, \phi_1, S_1), \fac(G_1, Z, S_1)) \in \sig(G_1, Z, \phi_1)\text{.}
  \]
  Since~$\sig(G_1, Z, \phi_1) = \sig(G_2, Z, \phi_2)$, there is a set $S_2 \subseteq V(G_2) \setminus Z$ with~$|S_1| = |S_2|$, $\req(G_1, Z, \phi_1, S_1) = \req(G_2, Z, \phi_2, S_2)$ and $\fac(G_1, Z, S_1) = \fac(G_2, Z, S_2)$. We claim that~$S' := (S \setminus S_1) \cup S_2$ is a vector connectivity set for~$G_2 \oplus \hat{G}$ with respect to~$\phi_2 \cup \hat{\phi}$; clearly~$|S| = |S'|$.

  Let us check that this is true indeed. We consider vertices in $D(G_2, \phi_2)$, $D(\hat{G}[Z], \hat{\phi})$, and vertices in $D(\hat{G} - Z, \hat{\phi})$ individually. Let us start with $D(G_2, \phi_2)$. For each vertex~$v_2 \in D(G_2, \phi_2)$, there is at least one vertex~$v_1 \in D(G_1, \phi_1)$ with~$\satcon(G_1, Z, \phi_1, S_1, v_1) = \satcon(G_2, Z, \phi_2(v_2), S_2, v_2)$. Let us show that the demand of~$v_2$ is satisfied in~$G_2 \oplus \hat{G}$ by $S'$.

  Consider first the case that $v_1 \in S_1$. Then $(A, B, C)$ is a satisfying connector for $v_1$ in $G_1$ for any three sets $A, B, C \subseteq Z$, mutually disjoint. In particular $(\emptyset, \emptyset, Z)$ is a satisfying connector for~$v_1$ and since the set of satisfying connectors for $v_1$ equals the one for $v_2$, $(\emptyset, \emptyset, Z)$ is a satisfying connector for~$v_2$. By the definition of satisfying connector either $v_2 \in S_2$, or there is a $(v_2, \phi_2(v_2), \emptyset, S_2)$-constrained path packing in $G_2 - Z$, meaning that there are $\phi_2(v_2)$ \indep{v_2} paths from $v_2$ to $S_2$ in $G_2 - Z$. Hence, the demand of $v_2$ is satisfied if $v_1 \in S_1$.

  Now assume that $v_1 \notin S_1$ and observe that $(V(G_1), V(\hat{G}))$ is a separation of $\hat{G} \oplus G_1$. Since there are $\phi_1(v_1)$ \indep{v_1} paths from $v_1$ to~$S$ in $\hat{G} \oplus G_1$, by Lemma~\ref{lem:path-packings-over-separator} and since $v_1 \in V(G_1) \setminus V(\hat{G})$, there is a partition of $Z$ into four sets $A, B, C, D$ such that there is a $(v_1, \phi_1(v_1), A, B \cup (S \setminus V(\hat{G})))$-constrained path packing~$\P_1$ in $G_1 - C$ and a $(B, A \cup (S \cap V(\hat{G})))$-constrained path packing~$\hat{\P}$ in~$\hat{G} - D$. Because $S \setminus V(\hat{G}) = S_1$, packing $\P_1$ is also $(v_1, \phi_1(v_1), A, B \cup S_1)$-constrained. Thus $\P_1$ witnesses that $(A, B, C) \in \satcon(G_1, Z, \phi_1(v_1), S_1, v_1)$, meaning that also $(A, B, C) \in \satcon(G_2, Z, \phi_2(v_2), S_2, v_2)$. Applying the definition of satisfying connector again, we have a $(v_2, \phi_2(v_2), A, B \cup S_2)$-constrained path packing $\P_2$ in $G_2 - C$. Note that $\P_2$ is also $(v_2, \phi_2(v_2), A, B \cup (S' \setminus V(\hat{G}))$-constrained. Applying again Lemma~\ref{lem:path-packings-over-separator}, the two path packings $\P_2$ and $\hat{\P}$ thus witness that there are $\phi(v_2)$ \indep{v_2} paths from $v$ to $S$ in $\hat{G} \oplus G_2$.

  Next, consider $v \in D(\hat{G} - Z, \hat{\phi})$; there are  $\hat{\phi}(v)$ \indep{v} paths from $v$ to~$S$ in $\hat{G} \oplus G_1$. Applying Lemma~\ref{lem:path-packings-over-separator} with the separation $(V(G_1), V(\hat{G}))$, we obtain a partition of $Z$ into $A, B, C, D$ (since $v \in V(\hat{G}) \setminus V(G_1)$), a $(A, B \cup S_1)$-constrained path packing~$\P_1$ in $G_1 - C$ and a $(v, \hat{\phi}(v), B, A \cup (S \cap V(\hat{G}))$-constrained path~$\hat{\P}$ packing in $\hat{G} - D$. By the definition of provided connector, we have $(A, B, C) \in \fac(G_1, Z, S_1) = \fac(G_2, Z, S_2)$. Thus, again by the definition of provided connector, there is a $(A, B \cup S_2)$-constrained path packing in~$G_2 - C$. Applying again Lemma~\ref{lem:path-packings-over-separator}, the packings $\P_2$ and $\hat{\P}$ witness that there are $\hat{\phi}(v)$ \indep{v} paths from $v$ to $S'$ in $\hat{G} \oplus G_2$.

  Finally, consider $v \in D(\hat{G}[Z]), \hat{\phi})$; there are again  $\hat{\phi}(v)$ \indep{v} paths from $v$ to~$S$ in $\hat{G} \oplus G_1$. Now we apply Lemma~\ref{lem:path-packings-over-separator} with the separation $(V(G_1), V(\hat{G}))$. This time, we obtain a partition of $Z \setminus \{v\}$ into $A, B, C, D$, and an integer $i$ together with a $(v, i, A, B \cup S_1)$-constrained path packing~$\P_1$ in~$G_1 - C$ and a $(v, \hat{\phi}(v) - i, B, A \cup (S \cap V(\hat{G})))$-constrained path packing~$\hat{\P}$ in~$\hat{G} - D$. By the definition of provided special connector, the packing $\P_1$ witnesses that $(v, i, A, B, C) \in \fac(G_1, Z, S_1) = \fac(G_2, Z, S_2)$. Hence, again by this definition, there is also a $(v, i, A, B \cup S_2)$-constrained path packing~$\P_2$ in $G_2 - C$. Applying Lemma~\ref{lem:path-packings-over-separator} again, the packings $\P_2$ and $\hat{\P}$ witness that there are $\hat{\phi}(v)$ \indep{v} paths from $v$ to $S'$ in $\hat{G} \oplus G_2$. 

  Overall we showed that each vertex~$v$ with nonzero demand~$(\hat{\phi} \cup \phi_2)(v)$ has as many \indep{v} paths from $v$ to $S'$ in $\hat{G} \oplus G_2$, meaning that $S'$ is a vector connectivity set. Since $|S'| = |S|$, this shows that Rule~\ref{rule:replace-Y} is safe.\qed
\end{longproof}

\paragraph{Putting things together.} 
We can now state our kernelization procedure for instances $(G,\phi,k)$ of \VdConk. The only missing piece is to argue why and how we may reduce vertices in $G$ that are not contained in any set of $\Y(G,\phi,d)$.

\begin{theorem}\label{theorem:vdconk:vertexlinearkernel}
\VdConk has a vertex-linear polynomial kernelization.
\end{theorem}

\begin{proof}
Given an instance $(G,\phi,k)$ of \VdConk the kernelization proceeds as follows. Throughout, we refer to the current instance by $(G,\phi,k)$ and recall the use of $D=\{v\in V(G)\mid \phi(v)\geq 1\}$.
\begin{enumerate}
 \item Apply Rule~\ref{rule:reducedemand} exhaustively and then apply Rule~\ref{rule:bounddemand} (this may return answer \no if we have more than $d^2k$ demand vertices).\label{kernelization:step:one}
 \item Apply Rule~\ref{rule:replace-Y} once if possible.
 \item Return to Step~\ref{kernelization:step:one} if Rule~\ref{rule:replace-Y} was applied.
 \item Let $W:=D\cup\bigcup_{Y\in\Y}N[Y]$. Perform the torso operation on $W$ in $G$ to obtain $G'$. That is, carry out the following steps:
 \begin{enumerate}
  \item Start with $G'=G[W]$.
  \item For every pair $u,v\in W$, if there is a $u,v$-path in $G$ with internal vertices from $V\setminus W$ then add the edge $\{u,v\}$ to $G'$.
 \end{enumerate}
 \item Return $(G',\phi',k)$ as the kernelized instance, where $\phi'=\phi|_W$ is $\phi$ restricted to $W$.
\end{enumerate}

\emph{Correctness.} We already know that Rules~\ref{rule:reducedemand} through \ref{rule:replace-Y} are correct; it remains to discuss the effect of the torso operation: Proposition~\ref{prop:HSequiv} implies that minimal solutions $S$ for $(G,\phi,k)$ are completely contained in the union of sets $X\in\X$, since only such vertices can contribute to $S$ being a hitting set for \X. It follows, by Lemma~\ref{lemma:xiny}, that every minimal solution $S$ is also contained in $W$. Thus, if before the torso operation every vertex $v\in D$ has $\phi(v)$ paths to $S$ then the same is true after the operation since there are shortcut edges for all paths with internal vertices from $V\setminus W$. The converse is more interesting.

Assume that $(G,\phi,k)$ is \no and fix an arbitrary set $S\subseteq W=V(G')\subseteq V(G)$ of size at most $k$; we will show that $S$ is not a solution for $(G',\phi',k)$. By assumption $S$ is not a solution for $(G,\phi,k)$ and, by Proposition~\ref{prop:HSequiv}, it is not a hitting set for $\X=\X(G,\phi)$. Accordingly, fix a set $X\in\X$ with $S\cap X=\emptyset$. By definition of $\X$, let $v\in X$ with $\phi(v)>|N(X)|$. By Lemma~\ref{lemma:xiny}, there is a $Y\in\Y(G,\phi,d)$ with $X\subseteq Y$. Thus,~$N[X]\subseteq N[Y]\subseteq W$. If there were $\phi(v)$ paths from $v$ to $S$ in $G'$ then at least one of them avoids $N(X)$, since $|N(X)|<\phi(v)$; let $P'$ denote a path from $v$ to $S$ in $G'$ that avoids $N(X)$. Undoing the torso operation, we get a walk $P$ in $G$, with additional interval vertices from $V(G)\setminus W$. Since $(V(G)\setminus W)\cap N(X)=\emptyset$, this walk also avoids $N(X)$ and implies that $X$ contains at least one vertex of $S$; a contradiction to $S\cap X=\emptyset$. Thus, no $S\subseteq V(G')$ of size at most $k$ is a solution for $(G',\phi',k)$, implying that $(G',\phi',k)$ is \no, as required.

\emph{Size.} The output graph $G'$ has vertex set $W=D\cup \bigcup_{Y\in\Y}N[Y]$ where $D$ and $\Y$ correspond to a fully reduced instance $(G,\phi,k)$. By Rule~\ref{rule:bounddemand} we have $|D|\leq d^2k$. By Lemma~\ref{lemma:boundy} the set $\Y$ contains at most $2^{\demandiny+d}d^2k$ sets, each of size bounded by $f(d)$ for some computable function depending only on~$d$. By Definition~\ref{definition:y} the neighborhood $N(Y)$ of each set $Y$ has size at most $d$. It follows that $G'$ has $\Oh(k)$ vertices, as claimed. Clearly, the total encoding size for an instance can be bounded by $\Oh(k^2)$ since $d$ is a constant.

\emph{Runtime.} By Lemma~\ref{lemma:reducedemand-poly}, Rule~\ref{rule:reducedemand} can be applied exhaustively in polynomial time. Clearly, Rule~\ref{rule:bounddemand} can be applied in polynomial time as it only checks the number of nonzero-demand vertices. Finding one application of Rule~\ref{rule:replace-Y} can be done by iterating over $Y\in\Y$ and applying Lemma~\ref{lemma:replacement-Y} to each $Y$ until we find a replacement subgraph that is strictly smaller; in total this takes polynomial time. Furthermore, repeating these steps whenever Rule~\ref{rule:replace-Y} has been applied gives only a polynomial factor because each time the instance shrinks by at least one vertex. Finally, it is easy to implement the torso operation in polynomial time.\qed
\end{proof}

\section{Kernelization lower bound}\label{section:kernellowerbound}

In this section, we prove that \VConk admits no polynomial kernelization unless \containment. We give a reduction from \HSm, i.e., \HS with parameter number of sets, which also makes a polynomial Turing kernelization unlikely (cf.~\cite{HermelinKSWW13}). Since demands greater than $k+1$ can be safely replaced by demand $k+1$, implying $d\leq k+1$, the lower bound applies also to parameterization by $d+k$.

\begin{theorem}
\VConk does not admit a polynomial kernelization unless \containment and the polynomial hierarchy collapses.
\end{theorem}

\begin{shortproof}
We give a polynomial parameter transformation from \HSm to \VConk, which is known to imply the claimed lower bound (cf.~\cite{BodlaenderJK14}). Let $(U,\F,k)$ be an instance of \HSm with parameter $m=|\F|$; w.l.o.g.~$k\leq m$. Let $n:=|U|$. We construct a graph $G$ on $2(k+1)m+n$ vertices that has a vector connectivity set of size at most $k'=(k+1)m+k=\Oh(m^2)$ if and only if $(U,\F,k)$ is \yes for \HSm. We will only show the construction and defer the correctness proof to the full version.

\emph{Construction.}
Make one vertex $x_u$ for each element $u\in U$, and make $2(k+1)$ vertices $y_{1,F},\ldots,y_{k+1,F},y'_{1,F},\ldots,y'_{k+1,F}$ for each set $F\in\F$. We add the following edges:
\begin{enumerate}
 \item Add $\{y_{i,F},y'_{i,F}\}$ for all $i\in\{1,\ldots,k+1\}$ and $F\in\F$.
 \item Add $\{x_{u},y_{i,F}\}$ for all $i\in\{1,\ldots,k+1\}$, $F\in\F$, and $u\in F$.
 \item Make the set of all vertices $y_{i,F}$ a clique (not including any $y'$-vertex).
\end{enumerate}
Set the demand $\phi$ of each $y'_{i,F}$ vertex to $2$ and of each $y_{i,F}$ vertex to $(k+1)m+1$; all $x$-vertices have demand zero. Set the budget $k'$ to $(k+1)m+k$. This completes the construction of an instance $(G,\phi,k')$, which can be easily performed in polynomial time.\qed
\end{shortproof}

\begin{longproof}
We give a polynomial parameter transformation from \HSm to \VConk, which is known to imply the claimed lower bound (cf.~\cite{BodlaenderJK14}). Let $(U,\F,k)$ be an instance of \HSm with parameter $m=|\F|$; w.l.o.g.~$k\leq m$. Let $n:=|U|$. We construct a graph $G$ on $2(k+1)m+n$ vertices that has a vector connectivity set of size at most $k'=(k+1)m+k=\Oh(m^2)$ if and only if $(U,\F,k)$ is \yes for \HSm.

\emph{Construction.}
Make one vertex $x_u$ for each element $u\in U$, and make $2(k+1)$ vertices $y_{1,F},\ldots,y_{k+1,F},y'_{1,F},\ldots,y'_{k+1,F}$ for each set $F\in\F$. We add the following edges:
\begin{enumerate}
 \item Add $\{y_{i,F},y'_{i,F}\}$ for all $i\in\{1,\ldots,k+1\}$ and $F\in\F$.
 \item Add $\{x_{u},y_{i,F}\}$ for all $i\in\{1,\ldots,k+1\}$, $F\in\F$, and $u\in F$.
 \item Make the set of all vertices $y_{i,F}$ a clique (not including any $y'$-vertex).
\end{enumerate}
Set the demand $\phi$ of each $y'_{i,F}$ vertex to $2$ and of each $y_{i,F}$ vertex to $(k+1)m+1$; all $x$-vertices have demand zero. Set the budget $k'$ to $(k+1)m+k$. This completes the construction of an instance $(G,\phi,k')$, which can be easily performed in polynomial time.

\emph{Correctness.} Assume first that $(G,\phi,k')$ is \yes and let $S$ a vector connectivity set of size at most $k'$. Note that $S$ must contain all vertices $y'_{i,F}$ since they have demand of $2$ but only one neighbor (namely $y_{i,F}$). This accounts for $(k+1)m$ vertices in $S$; there are at most $k$ further vertices in $S$. Let $T$ contain exactly those elements $u\in U$ such that $x_u\in S$; thus $|T|\leq k$. We claim that $T$ is a hitting set for $\F$.
Let $F\in\F$ and assume that $T\cap F=\emptyset$. It follows that $S$ contains no vertex $x_u$ with $u\in F$. Since at most $k$ vertices in $S$ are not $y'$-vertices, we can choose $i\in\{1,\ldots,k+1\}$ such that $S$ does not contain $y_{i,F}$. Consider the set $C$ consisting of all $y$-vertices other than $y_{i,F}$ as well as the vertex $y'_{i,F}$. In $G-C$ we find a connected component containing $y_{i,F}$ and all $x_u$ with $u\in F$ but no further vertices. Crucially, all other neighbors of $y_{i,F}$ are $y'_{i,F}$ and all $y$-vertices, and $x$-vertices only have $y$-vertices as neighbors. By assumption $S$ contains no vertex of this connected component. This yields a contradiction cause $C$ is of size $(k+1)m$ and separates $y_{i,F}$ from $S$, but since $S$ is a solution with $y_{i,F}\notin S$ there should be $(k+1)m+1$ disjoint paths from $y_{i,F}$ to $S$. Thus, $S$ must contain some $x_u$ with $u\in F$, and then $T\cap F\neq\emptyset$.

Now, assume that $(U,\F,k)$ is \yes for \HSm and let $T$ a hitting set of size at most $k$ for $\F$. We create a vector connectivity set $S$ by selecting all $x_u$ with $u\in T$ as well as all $y'$-vertices; thus $|S|\leq k'=(k+1)m+k$. Clearly, this satisfies all $y'$-vertices. Consider any vertex $y_{i,F}$ and recall that its demand is $\phi(y_{i,F})=(k+1)m+1$. We know that $S$ contains at least one vertex $x_u$ with $u\in F$ that is adjacent to $y_{i,F}$. Thus, we can find the required $(k+1)m+1$ disjoint paths from $y_{i,F}$ to $S$:
\begin{itemize}
 \item We have one path $(y_{i,F},y'_{i,F})$ and one path $(y_{i,F},x_u)$.
 \item For all $(j,F')\neq (i,F)$ we get one path $(y_{i,F},y_{j,F},y'_{j,F})$; we get $(k+1)m-1$ paths total.
\end{itemize}
It follows that $(G,\phi,k')$ is \yes for \VConk.

We have given a polynomial parameter transformation from \HSm, which is known not to admit a polynomial kernelization unless \containment \cite{DomLS09} (see also \cite{HermelinKSWW13}). This is known to imply the same lower bound for \VConk~\cite{BodlaenderJK14}.\qed
\end{longproof}

\section{Conclusion} \label{section:conclusion}

We have presented kernelization and approximation results for \VCon and \VdCon. An important ingredient of our results is a reduction rule that reduces the number of vertices with nonzero demand to at most $d^2\opt$ (or, similarly, to at most $\opt^3+\opt$ or $k^3+k$). From this, one directly gets approximation algorithms with ratios $d^2$ and $(\opt^2+1)$; we improved these to factors $d$ and $\opt$, respectively, by a local-ratio type algorithm. Recall that \VdCon is \APX-hard already for $d=4$~\cite{CicaleseMR14}.

On the kernelization side we show that \VConk does not admit a polynomial kernelization unless \containment. Since demands greater than $k+1$ can be safely replaced by demand $k+1$ (because they cannot be fulfilled without putting the vertex into the solution) the lower bound extends also to parameter $k+d$. For \VdConk, where $d$ is a problem-specific constant, we give an explicit vertex-linear kernelization with at most $f(d)\cdot k=\Oh(k)$ vertices; the computable function $f(d)$ is superpolynomial in $d$, which is necessary (unless \containment) due to the lower bound for $d+k$.

Finally, the reduction to $k^3+k$ nonzero demand vertices allows an alternative proof for fixed-parameter tractability: We give a randomized FPT-algorithm for \VConk that finds a solution by seeking a set of size $k$ that is simultaneously independent in each of $\ell=k^3+k$ linear matroids, each of which handles one demand vertex; for this we use an algorithm of Marx~\cite{Marx09} for linear matroid intersection, which is fixed-parameter tractable in $k+\ell=\Oh(k^3)$.

\paragraph{Acknowledgments.} The authors are grateful to anonymous reviewers for helpful comments improving the presentation and technical content of the paper. In particular, a reviewer commented on the definition of signatures, leading to a significantly simpler way of computing them; other reviewers pointed out that a non-constructive kernelization can be obtained from the literature on meta kernelization (in particular, from Fomin et al.~\cite{FominLST13}). 

Stefan Kratsch was supported by the German Research Foundation (DFG), KR 4286/1. Manuel Sorge was also supported by the German Research Foundation (DFG), project DAPA, NI 369/12.

\bibliographystyle{abbrv}
\bibliography{vectorconnectivity}

\end{document}